\documentclass[11pt,letterpaper,accepted=2021-08-16,onecolumn]{quantumarticle}
\pdfoutput=1
\usepackage[utf8]{inputenc}
\usepackage[english]{babel}
\usepackage{csquotes}
\usepackage[T1]{fontenc}
\usepackage{lastpage}
\usepackage{enumerate}
\usepackage{fancyhdr}
\usepackage{mathrsfs}
\usepackage{xcolor}
\usepackage{graphicx}
\usepackage{listings}
\usepackage{amssymb,amsmath,amsthm}
\usepackage{fullpage}
\usepackage{tikz}
\usepackage[font=small,labelfont=bf]{caption}
\usepackage{subcaption}

\usepackage{authblk}
\usepackage[%
	backend=biber,
	style=alphabetic,
	date=year,
	sorting=nyt,
	isbn=false,
	maxbibnames=20,
	mincrossrefs=99, % Do not list crossref'd entries.
	giveninits=true % only use initials
]{biblatex}
\DeclareFieldFormat
  [article,inbook,incollection,inproceedings,patent,thesis,unpublished]
  {titlecase:title}{\MakeSentenceCase*{#1}}
\DeclareSourcemap{%
	\maps[datatype=bibtex]{%
		\map[overwrite]{%
			\step[fieldsource=doi, final]
			\step[fieldset=url, null]
			\step[fieldset=eprint, null]
		}
	}
}

\DeclareSourcemap{
  \maps[datatype=bibtex, overwrite]{
    \map{
      \step[fieldset=address, null]
      \step[fieldset=location, null]
      \step[fieldset=venue, null]
      \step[fieldset=eventdate, null]
      \step[fieldset=eventtitle, null]
      \step[fieldset=editor, null]
    }
  }
}

\usepackage[lined,boxed,algosection,linesnumbered,noend]{algorithm2e}
\SetAlCapSkip{0.5em}
\SetAlCapNameFnt{\small}
\SetAlCapFnt{\small}
\SetEndCharOfAlgoLine{} % Remove semicolons at end of line
\SetKwInOut{Data}{Data}
\SetKwInOut{Input}{Input}
\SetKwInOut{Output}{Output}
\SetKwFor{For}{for}{\string:}{endfor}
\SetKwFor{While}{while}{\string:}{endw}
\SetKwFor{ForEach}{foreach}{\string:}{endfch}
\SetKwFor{ForAll}{forall}{\string:}{endfall}
\SetKwProg{Fn}{function}{\string:}{endfunction}
\newcommand\doubleplus{\ensuremath{\mathbin{+\mkern-10mu+}}}
\usepackage{listings}

\usepackage{mathtools}
\DeclareMathOperator{\revop}{R}

\DeclareMathOperator{\dist}{dist}

\DeclareMathOperator{\rt}{rt}

\DeclareMathOperator{\EX}{\mathbb{E}}

\DeclarePairedDelimiter\floor{\lfloor}{\rfloor}

\DeclarePairedDelimiter\set{\{}{\}}
\DeclarePairedDelimiter\abs{\lvert}{\rvert}

\DeclarePairedDelimiterXPP\rev[1]{\revop}{(}{)}{}{#1}
\DeclarePairedDelimiterXPP\bigo[1]{O}{(}{)}{}{#1}
\DeclarePairedDelimiterXPP\tildebigo[1]{\widetilde O}{(}{)}{}{#1}
\DeclarePairedDelimiterXPP\rnumber[1]{\rt}{(}{)}{}{#1}
\DeclarePairedDelimiterXPP\expected[1]{\EX}{[}{]}{}{#1}
\DeclarePairedDelimiterXPP\probability[1]{\Pr}{(}{)}{}{#1}

\newcommand{\N}{\mathbb{N}}

\newcommand{\e}{\varepsilon}
\newcommand{\underwrite}[3][]{% \underwrite[<thickness>]{<numerator>}{<denominator>}
  \genfrac{}{}{#1}{}{\textstyle #2}{\textstyle #3}
}
\newcommand{\swap}[0]{\textsc{swap}}

\newcommand{\cz}[0]{\textsc{cz}}
\newcommand{\RoutingViaMatchings}[0]{\textsc{Routing via Matchings}}

\newcommand{\bullets}{****}

\usepackage{color}
\usepackage{xcolor}
\definecolor{dark-red}{rgb}{0.4,0.15,0.15}
\definecolor{dark-blue}{rgb}{0.15,0.15,0.4}
\definecolor{medium-blue}{rgb}{0,0,0.5}

\definecolor{mycomment}{rgb}{0.3,0.7,0.8}
\definecolor{mygray}{rgb}{0.5,0.5,0.5}
\definecolor{lightgray}{rgb}{0.95,0.95,0.95}
\definecolor{mymauve}{rgb}{0.58,0,0.82}

\usepackage{hyperref}
\hypersetup{%
	breaklinks=true,
	colorlinks,
	linkcolor={dark-blue},
	citecolor={dark-blue},
	urlcolor={medium-blue},
	pdfpagemode=UseOutlines,
	pdfborder={0 0 0}
}

\newcommand{\newtext}[1]{#1}

\newcommand{\length}[1]{|#1|}

\newcommand{\Kfive}{                    \begin{tikzpicture}
            \draw[fill=black] (0.5,0) circle (3pt);
            \draw[fill=black] (2.5,0) circle (3pt);
            \draw[fill=black] (3,2) circle (3pt);
            \draw[fill=black] (1.5,3) circle (3pt);
            \draw[fill=black] (0,2) circle (3pt);
            \draw[] (0.5,0) -- (2.5,0) -- (3,2) -- (1.5,3) -- (0,2) -- (0.5,0) -- (3,2) -- (0,2) -- (2.5,0) -- (1.5,3) -- (0.5,0);

        \end{tikzpicture}
}

\usepackage{enumitem}
\newlist{ienumerate}{enumerate*}{1}
\setlist*[ienumerate,1]{%
	label=(\roman*),
}

\usepackage[nameinlink,capitalise]{cleveref}
\crefname{figure}{Figure}{Figures}

\newtheorem{theorem}{Theorem}[section]
\newtheorem{lemma}[theorem]{Lemma}
\newtheorem{corollary}[theorem]{Corollary}
\theoremstyle{definition}

\title{Quantum routing with fast reversals}

\author[1,4]{Aniruddha Bapat}
\orcid{0000-0001-9489-9165}
\email{ani@umd.edu}
\author[1,2,3]{Andrew M. Childs}
\orcid{0000-0002-9903-837X}
\email{amchilds@umd.edu}
\author[1,4]{Alexey V. Gorshkov}
\orcid{0000-0003-0509-3421}
\email{gorshkov@umd.edu}
\author[5]{Samuel King}
\orcid{0000-0002-9575-9225}
\author[1,2,3]{Eddie Schoute}
\orcid{0000-0002-5613-1443}
\email{eschoute@umd.edu}
\author[6]{Hrishee Shastri}
\orcid{0000-0002-6188-1572}
\affil[1]{Joint Center for Quantum Information and Computer Science, NIST\!/University of Maryland, College Park, Maryland 20742, USA}
\affil[2]{Institute for Advanced Computer Studies, University of Maryland, College Park, Maryland 20742, USA}
\affil[3]{Department of Computer Science, University of Maryland, College Park, Maryland 20742, USA}
\affil[4]{Joint Quantum Institute, NIST\!/University of Maryland, College Park, Maryland 20742, USA}
\affil[5]{University of Rochester, Rochester, New York 14627, USA}
\affil[6]{Reed College, Portland, Oregon 97202, USA}

\date{}
\begin{document}

%%%%%%%%%%%%%%%%%%%%%%%%%%%%%%%%%%%%%%%%%%%%%%%%%%%%%%%%%%%%%%%%%%%%%%%%%%%%%%%%

\maketitle

\begin{abstract}
\noindent We present methods for implementing arbitrary permutations of qubits 
under interaction constraints. Our protocols make use of previous methods for rapidly reversing the order of qubits along a path.
Given nearest-neighbor interactions on a path of length $n$,
we show that there exists a constant $\epsilon \approx 0.034$ such that the quantum routing time
is at most $(1-\epsilon)n$,
whereas any \swap{}-based protocol needs at least time $n-1$.
This represents the first known quantum advantage over \swap{}-based routing methods
and also gives improved quantum routing times for realistic architectures such as grids.
Furthermore, we show that our algorithm approaches a quantum routing time of $2n/3$ in expectation for uniformly random permutations,
whereas \swap{}-based protocols require time $n$ asymptotically.
Additionally, we consider sparse permutations that route $k \le n$ qubits
and give algorithms with quantum routing time at most $n/3 + \bigo{k^2}$ on paths and at most 
$2r/3 + \bigo{k^2}$ on general graphs with radius $r$.
\end{abstract}

%%%%%%%%%%%%%%%%%%%%%%%%%%%%%%%%%%%%%%%%%%%%%%%%%%%%%%%%%%%%%%%%%%%%%%%%%%%%%%%%
\section{Introduction}
Qubit connectivity limits quantum information transfer,
which is a fundamental task for %both 
quantum computing.
%and long-range quantum communication.
While the common model for quantum computation usually assumes all-to-all connectivity,
proposals for scalable quantum architectures do not have this capability~\cite{Monroe2013,Monroe2014,Brecht2016}.
Instead, quantum devices arrange qubits in a fixed architecture that fits within engineering and design constraints.
For example, the architecture may be grid-like~\cite{McClure2019,Arute2019} or consist of a network of submodules~\cite{Monroe2013,Monroe2014}.
Circuits that assume all-to-all qubit connectivity can be mapped onto these architectures via protocols for \emph{routing} qubits,
i.e., permuting them within the architecture using local operations.
% Likewise, large-scale quantum networks have qubit connectivity limited by cost or ability to set up links over long distances and the problem of fast routing is key to the efficiency of quantum communication protocols on such networks.
% (In such networking scenarios we do not expect the underlying state reversal protocol can be implemented, however.)

Long-distance gates can be implemented using \swap{} gates along edges of the graph of available interactions.
A typical procedure swaps pairs of distant qubits along edges until they are adjacent,
at which point the desired two-qubit gate is applied on the target qubits.
These swap subroutines can be sped up by parallelism and careful scheduling~\cite{Saeedi2011,Shafaei2013,Shafaei2014,Pedram2016,Lye2015,Murali2019,Zulehner2019}.
Minimizing the \swap{} circuit depth corresponds to the \RoutingViaMatchings{} problem~\cite{Alon1994,Childs2019}.
The minimal \swap{} circuit depth to implement any permutation on a graph $G$ is given by
its \emph{routing number}, $\rnumber{G}$~\cite{Alon1994}.
Deciding $\rnumber{G}$ 
% when $\rnumber{G} > 2$ 
is generally NP-hard~\cite{Banerjee2017},
but there exist algorithms for architectures of interest such as grids and other graph products~\cite{Alon1994,Zhang1999,Childs2019}. Furthermore, one can establish lower bounds on the routing number as a function of graph diameter and other properties.

% Given ancillary resources and the ability to perform fast local measurements and classical communication, qubit connectivity can be enhanced by, e.g., pre-shared long-distance entanglement in the form of virtual quantum links~\cite{Schoute2016}.

Routing using \swap{} gates does not necessarily give minimal circuit evolution time since it is effectively classical and does not make use of the full power of quantum operations.
Indeed, faster protocols are already known for specific permutations in specific qubit geometries such as the path~\cite{Raussendorf2005,spin-chain}.
These protocols tend to be carefully engineered and do not generalize readily to other permutations, leaving open the general question of devising faster-than-\swap{} quantum routing.
In this paper, we give a positive answer to this question.

\newtext{Following \cite{Raussendorf2005,spin-chain}, we consider a continuous-time model of routing, where the protocol is defined by a Hamiltonian that can only include nearest-neighbor interactions.
To make consistent comparisons with a gate-based model of routing, we bound the spectral norm of interactions~\cite{spin-chain} so that a \swap{} gate takes unit time~\cite{Vidal2002}, as determined by the canonical form of a two-qubit Hamiltonian~\cite{Bennett2002}.
We suppose that single-qubit operations can be performed arbitrarily fast, a common assumption~\cite{Vidal2002,Bennett2002} that is practically well-motivated due to the relative ease of implementing single-qubit rotations.}

Rather than directly engineering a quantum routing protocol,
% that works for any permutation, 
we consider a hybrid strategy that leverages a known protocol for quickly performing a specific permutation to implement general quantum routing.
Specifically, we consider the reversal operation 
\begin{equation}
  \label{eq:SWAPrb}
  \rho \coloneqq \prod\limits_{k=1}^{\lfloor\frac{n}{2}\rfloor}\swap_{k,n+1-k}
\end{equation}
that swaps the positions of qubits about the center of a length-$n$ path.
Fast quantum reversal protocols are known in the gate-based~\cite{Raussendorf2005} and time-independent Hamiltonian~\cite{spin-chain} settings.
The reversal operation can be implemented in time~\cite{spin-chain}
\begin{equation}
    \label{eq:COSTr}
    T(\rho) \le \frac{\sqrt{(n+1)^{2} - p(n)}}{3} \leq \frac{n+1}{3},
\end{equation}
where $p(n) \in \{0,1\}$ is the parity of $n$.
Both protocols exhibit an asymptotic time scaling of $n/3 + \bigo{1}$,
which is asymptotically three times faster than the best possible \swap{}-based time of $n-1$ (bounded by the diameter of the graph)~\cite{Alon1994}. The odd-even sort algorithm provides a nearly tight time upper bound of $n$~\cite{Lakshmi84} and will be our main point of comparison.

\newtext{The Hamiltonian protocol of~\cite{spin-chain} can be understood by looking at the time evolution of the site Majorana operators obtained by a Jordan-Wigner transformation of the spin chain. In this picture, the protocol can be interpreted as the rotation of a fictitious particle of spin $n+1/2$ whose magnetization components are in one-to-one correspondence with the Majoranas on the chain. A reversal corresponds to a rotation of the large spin by an angle of $\pi$. 
The gate-based reversal protocol~\cite{Raussendorf2005} is a special case of a quantum cellular automaton with a transition function given by the $(n+1)$-fold product of nearest-neighbor controlled-Z (\cz{}) operations---an operation that can be done 3 times faster than a \swap{} gate---and Hadamard operations.
In an open spin chain, this process spreads out local Pauli observables at site $i$ over the chain and ``refocuses'' them at site $n+1-i$ in $n+1$ steps for every $i$.
The ability to spread local observables (which is present in the gate-based and Hamiltonian protocols but not in \swap{}-based protocols) may be key to obtaining a speedup over \swap{}-based algorithms.}

\newtext{We expect both the gate-based and Hamiltonian protocols to be implementable on near-term quantum devices.
The gate-based protocol uses nearest-neighbor \cz{} gates and Hadamard gates, both of which are widely used on existing quantum platforms.
The Hamiltonian protocol involves nearest-neighbor Pauli XX interactions with non-uniform couplings, which is within the capabilities of, e.g., superconducting architectures~\cite{kjaergaard2020}.}

Routing using reversals has been studied extensively due to its applications in comparative genomics
(where it is known as \emph{sorting by reversals})~\cite{Bafna93,Sankoff95}.
% \citeauthor{Caprara97}~\cite{Caprara97} determined that finding an optimal sorting reversal sequence is \textsc{np}-hard.
References~\cite{bender-bounds,pinter-skiena-genomic,Nguyen05} present routing algorithms where, much like in our case, reversals have length-weighted costs.
However, these models assume reversals are performed sequentially,
while we assume independent reversals can be performed in parallel,
where the total cost is given by the evolution time, akin to circuit depth.
To our knowledge, results from the sequential case are not easily adaptable to the parallel setting and require a different approach.

% Routing on paths has applications to more general graphs.
% Given any graph $G$, we can implement reversals in parallel by performing them along non-intersecting paths.  
Routing on paths is a fundamental building block for routing on more general graphs.
For example, a two-dimensional grid graph is the Cartesian product of two path graphs, 
and the best known routing routine applies a path routing subroutine 3 times~\cite{Alon1994}.
A quantum protocol for routing on the path in time $cn$, for a constant $c>0$, would imply a routing time of $3cn$ on the grid. \newtext{A similar speedup follows for higher-dimensional grids.}
More generally, routing algorithms for the \emph{generalized hierarchical product} of graphs can take advantage of faster routing of the path base graph~\cite{Childs2019}.
\newtext{For other graphs, it is open whether fast reversals can be used to give faster routing protocols for general permutations.}

In the rest of this paper, we present the following results on quantum routing using fast reversals.
In \cref{routing-first-examples}, we give basic examples of using fast reversals to perform routing on general graphs
to indicate the extent of possible speedup over \swap{}-based routing,
namely a graph for which routing can be sped up by a factor of $3$, and another for which no speedup is possible.
\cref{sparse-permutations} presents algorithms for routing sparse permutations,
where few qubits are routed,
both for paths and for more general graphs.
Here, we obtain the full factor 3 speedup over \swap{}-based routing.
Then, in \cref{results}, we prove the main result that 
there is a quantum routing algorithm for the path with worst-case constant-factor advantage over any \swap{}-based routing scheme. 
Finally, in \cref{averagecase}, we show that our algorithm has average-case routing time $2n/3 + o(n)$
and any \swap{}-based protocol has average-case routing time at least $n - o(n)$.

%%%%%%%%%%%%%%%%%%%%%%%%%%%%%%%%%%%%%%%%%%%%%%%%%%%%%%%%%%%%%%%%%%%%%%%%%%%%%%%%
\section{Simple bounds on routing using reversals}\label{routing-first-examples}

Given the ability to implement a fast reversal $\rho$ with cost given by~\cref{eq:COSTr}, the largest possible asymptotic speedup of reversal-based routing over \swap{}-based routing is a factor of $3$. 
This is because the reversal operation, which is a particular permutation, cannot be performed faster than $n/3 + o(n)$, and can be performed in time $n$ classically using odd-even sort. As we now show, some graphs can saturate the factor of $3$ speedup for general permutations, while other graphs do not admit any speedup over \swap{}s.

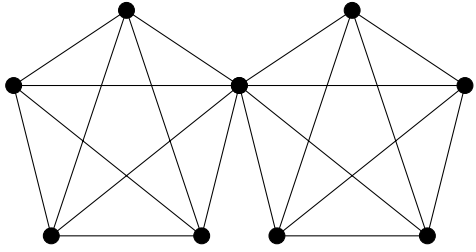
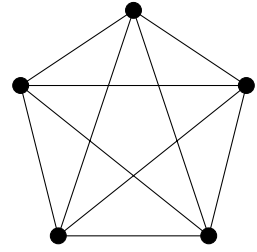
\begin{figure}
    \begin{subfigure}[b]{0.49\textwidth}
    \begin{center}
        \begin{tikzpicture}
            %% left subgraph 
            \draw[fill=black] (0.5,0) circle (3pt);
            \draw[fill=black] (2.5,0) circle (3pt);
            \draw[fill=black] (3,2) circle (3pt);
            \draw[fill=black] (1.5,3) circle (3pt);
            \draw[fill=black] (0,2) circle (3pt);
            %% vertex labels
            
            %%% edges
            \draw[] (0.5,0) -- (2.5,0) -- (3,2) -- (1.5,3) -- (0,2) -- (0.5,0) -- (3,2) -- (0,2) -- (2.5,0) -- (1.5,3) -- (0.5,0);

            % right subgraph
            %% vertices
            \draw[fill=black] (3.5,0) circle (3pt);
            \draw[fill=black] (5.5,0) circle (3pt);
            \draw[fill=black] (6,2) circle (3pt);
            \draw[fill=black] (4.5,3) circle (3pt);
            \draw[fill=black] (3,2) circle (3pt);
            %% vertex labels
            
            %%% edges
            \draw[] (3.5,0) -- (5.5,0) -- (6,2) -- (4.5,3) -- (3,2) -- (3.5,0) -- (6,2) -- (3,2) -- (5.5,0) -- (4.5,3) -- (3.5,0);
            
            % Bridge
            % \draw[] (3,2) -- (6,2);
        \end{tikzpicture}
        \end{center}
        \caption{Joined graph $K_{9}^*$.}\label{fig:joinedGraph}
        \end{subfigure}
    \begin{subfigure}[b]{0.49\textwidth}
    \begin{center}
        %\stargraph{7}{1.5}
        \Kfive
        \end{center}
        \caption{Complete graph $K_5$.}\label{fig:nospeedupgraph}
    \end{subfigure}
    \caption{%
        $K^{*}_{9}$ admits the full factor of $3$ speedup in the worst case when using reversals over \swap{}s,
        whereas $K_{5}$ admits no speedup when using reversals over \swap{}s.
    }
    \label{fig:extreme-graphs}
\end{figure}

\paragraph{Maximal speedup:} For $n$ odd, let $K^{*}_{n}$ denote two complete graphs, each on $(n+1)/2$ vertices, joined at a single ``junction" vertex for a total of $n$ vertices (\cref{fig:joinedGraph}). 
Consider a permutation on $K^{*}_{n}$ in which every vertex is sent to the other complete subgraph, except that the junction vertex is sent to itself.
To route with \swap{}s, note that each vertex (other than that at the junction) must be moved to the junction at least once, and only one vertex can be moved there at any time.
Because there are $(n+1)/2 - 1$ non-junction vertices on each subgraph,
implementing this permutation requires a \swap{}-circuit depth of at least $n - 1$.

On the other hand, any permutation on $K^{*}_{n}$ can be implemented in time $n/3 + O(1)$ using reversals.
First, perform a reversal on a path that connects all vertices with opposite-side destinations.  
After this reversal, every vertex is on the side of its destination and the remainder can be routed in at most 2 steps~\cite{Alon1994}. %since $\rt(K_{m})=2$ for all $m\in \mathbb N$
The total time is at most $(n+1)/3 + 2$, exhibiting the maximal speedup by an asymptotic factor of $3$.

\paragraph{No speedup:} Now, consider the complete graph on $n$ vertices, $K_n$ (\cref{fig:nospeedupgraph}).
Every permutation on $K_n$ can be routed in at most time 2 using \swap{}s~\cite{Alon1994}.
Consider implementing a 3-cycle on three vertices of $K_n$ for $n\geq 3$ using reversals. 
Any reversal sequence that implements this permutation will take at least time 2.
Therefore, no speedup is gained over \swap{}s in the worst case.

\bigskip

We have shown that there exists a family of graphs that allows a factor of $3$ speedup for any permutation when using fast reversals instead of \swap{}s, and others where reversals do not grant any improvement.
The question remains as to where the path graph lies on this spectrum. Faster routing on the path is especially desirable since this task is fundamental for routing in more complex graphs.

\section{An algorithm for sparse permutations} \label{sparse-permutations}

We now consider routing sparse permutations, where only a small number $k$ of qubits are to be moved.
For the path, we show that the routing time is at most $n/3 + \bigo{k^2}$. More generally, we show that for a graph of radius $r$, the routing time is at most $2r/3 + \bigo{k^2}$.
(Recall that the radius of a graph $G=(V,E)$ is $\min_{u \in V} \max_{v \in V} \dist(u,v)$, where $\dist(u,v)$ is the distance between $u$ and $v$ in $G$.)
Our approach to routing sparse permutations using reversals is based on the idea of bringing all $k$ qubits to be permuted to the center of the graph, rearranging them, and then sending them to their respective destinations.

\subsection{Paths}

A description of the algorithm on the path, called \texttt{MiddleExchange}, appears in \cref{alg:middle-exchange}. 
% Middle Exchange
    \begin{algorithm}[tbp]
        \Input{%
            $\pi$, a permutation
        }
        \SetKwFunction{FMidExch}{MiddleExchange}
        \Fn{\FMidExch{$\pi$}}{% 
            identify the labels $x_1, \dots, x_k \in [n]$ to be permuted, with $x_i < x_{i+1}$ \;
            let $t$ be the largest index for which $x_t \leq \floor{n / 2}$, i.e., the last label $x_t$ left of the median \;
            \For{$i = 1$ to $t-1$}{\label{me-begin-compression}
                perform $\rho(x_i-i+1, x_{i+1}-1)$
                % perform reversal from $(x_i-i+1)$ to $(x_{i+1}-1)$
                % join the labels $x_1, \dots, x_i$ to $x_{i+1}$ using at most one reversal $\rho_i$, not containing $x_{i+1}$
            }
            \For{$j = k$ to $t + 2$}{
                perform $\rho(x_j+k-j, x_{j-1}+1)$
                % perform reversal from $(x_j+k-j)$ to $(x_{j-1}+1)$
                % join the labels $x_j, \dots, x_k$ to $x_{j-1}$ using at most one reversal $\rho_j$, not containing $x_j$
            }
            perform $\rho(x_{t} - t + 1, \floor{n/2})$ \;
            % perform reversal from $(x_{t} - t + 1)$ to $\floor{n/2}$\;
            perform $\rho(x_{t+1} + k - t - 1, \floor{n/2}+1)$ \label{me-end-compression}\;
            % perform reversal from $(x_{t+1} + k - t - 1)$ to $\floor{n/2}+1$\label{me-end-compression}\;
            $\bar \rho \gets$ the sequence of all reversals so far\;
            % using at most two reversals $\rho_t$ and $\rho_{t+1}$, move the labels $x_1, \dots, x_t$ and $x_{t+1}, \dots, x_k$ to the median \label{me-end-compression}\;
            % $\Bar{\rho} := \rho_1, \dots, \rho_{t-1}, \rho_k, \dots, \rho_{t+2}, \rho_t, \rho_{t+1}$
            % \tcp*{The reversals performed so far, in order}
            route the labels $x_1, \dots, x_k$ such that after performing $\bar{\rho}$ in reverse order, each label is at its destination \label{me-inner-perm} \;
            perform $\bar \rho$ in reverse order \label{me-dilation}\;
        }
       \caption{MiddleExchange algorithm to sort sparse permutations on the path graph. We let $\rho(i,j)$ denote a reversal on the segment starting at $i$ and ending at $j$, inclusive.}\label{alg:middle-exchange}
    \end{algorithm}
\cref{fig:middle-exchange} presents an example of \texttt{MiddleExchange} for $k = 6$.
\begin{figure}
    \centering
        \begin{equation*}
    \begin{gathered}
        \bullets \underbracket{\texttt{5} \bullets} \texttt{3} \bullets \texttt{1} \bullets \bigg{|} \bullets \texttt{4} \bullets \texttt{6} \underbracket{\bullets \texttt{2}} \bullets \\
        \bullets \; \bullets \underbracket{\texttt{5} \: \texttt{3} \bullets} \texttt{1} \bullets \bigg{|} \bullets \texttt{4} \underbracket{\bullets \texttt{6} \: \texttt{2}} \bullets \; \bullets \\
        \bullets \; \bullets \; \bullets \underbracket{\texttt{3} \: \texttt{5} \: \texttt{1} \bullets} \bigg{|} \underbracket{\bullets \texttt{4} \: \texttt{2} \: \texttt{6}} \bullets \; \bullets \; \bullets \\
        \bullets \; \bullets \; \bullets \; \bullets \underbrace{\texttt{1} \: \texttt{5} \: \texttt{3} \: \bigg{|} \:  \texttt{6} \: \texttt{2} \: \texttt{4} }_\text{rearrange} \bullets \; \bullets \; \bullets \; \bullets \\
        \bullets \; \bullets \; \bullets \underbracket{\bullets \texttt{3} \: \texttt{1} \: \texttt{2}} \bigg{|} \underbracket{\texttt{5} \: \texttt{6} \: \texttt{4} \bullets} \bullets \; \bullets \; \bullets \\
        \bullets \; \bullets \underbracket{\bullets \texttt{2} \: \texttt{1}} \texttt{3} \bullets \bigg{|} \bullets \texttt{4} \underbracket{\texttt{6} \: \texttt{5} \bullets} \bullets \; \bullets \\
        \bullets \underbracket{\bullets \texttt{1}} \texttt{2} \bullets \texttt{3} \bullets \bigg{|} \bullets \texttt{4} \bullets \texttt{5} \underbracket{\texttt{6} \bullets} \bullets \\
        \bullets \texttt{1} \bullets \texttt{2} \bullets \texttt{3} \bullets \bigg{|} \bullets \texttt{4} \bullets \texttt{5} \bullets \texttt{6} \bullets
    \end{gathered}
    \end{equation*}
    \caption{Example of \texttt{MiddleExchange} (\cref{alg:middle-exchange}) on the path for $k = 6$.}
    \label{fig:middle-exchange}
\end{figure}

In \cref{sparse-lower-bound}, we prove that \cref{alg:middle-exchange} achieves a routing time of asymptotically $n/3$ when implementing a sparse permutation of $k = o(\sqrt{n})$ qubits on the path graph. 
First, let $\mathcal{S}_{n}$ denote the set of permutations on $\{1,\ldots,n\}$, so $|\mathcal{S}_{n}| = n!$.
Then, for any permutation $\pi \in \mathcal{S}_n$ that acts on a set of labels $\{1, \dots, n\}$, let $\pi_i$ denote the destination of label $i$ under $\pi$.
We may then write $\pi = (\pi_1, \pi_2, \dots, \pi_n)$. Let $\Bar{\rho}$ denote an ordered series of reversals $\rho_{1},\ldots,\rho_{m}$, and let $\Bar{\rho}_{1} \doubleplus \Bar{\rho}_{2}$ be the concatenation of two reversal series. Finally, let $S \cdot \rho$ and $S \cdot \Bar{\rho}$ denote the result of applying $\rho$ and $\Bar{\rho}$ to a sequence $S$, respectively, and let $\length{\rho}$ denote the length of the reversal $\rho$, i.e., the number of vertices it acts on.

\begin{theorem} \label{sparse-lower-bound}
    Let $\pi \in \mathcal{S}_n$ with $k = |\{x \in [n] \mid \pi_x \not= x\}|$ (i.e., $k$ elements are to be permuted, and $n-k$ elements begin at their destination).
    Then \cref{alg:middle-exchange} routes $\pi$ in time at most $n / 3 + \bigo{k^2}$.
\end{theorem}

\begin{proof}
    \cref{alg:middle-exchange} consists of three steps: compression (\cref{me-begin-compression}--\cref{me-end-compression}), inner permutation (\cref{me-inner-perm}), and dilation (\cref{me-dilation}). 
    Notice that compression and dilation are inverses of each other.
    
    Let us first show that \cref{alg:middle-exchange} routes $\pi$ correctly.
    Just as in the algorithm, let $x_1, \dots, x_k$ denote the labels $x \in [n]$ with $x_i < x_{i+1}$ such that $\pi_x \not= x$, that is, the elements that do not begin at their destination and need to be permuted.
    It is easy to see that these elements are permuted correctly:
    After compression, the inner permutation step routes $x_{i}$ to the current location of the label $\pi_{x_{i}}$ in the middle. 
    Because dilation is the inverse of compression, it will then route every $x_{i}$ to its correct destination.
    For the non-permuting labels, notice that they lie in the support of either no reversal or exactly two reversals, $\rho_{1}$ in the compression step and  $\rho_{2}$ in the dilation step.
    Therefore $\rho_{1}$ reverses the segment containing the label and $\rho_{2}$ re-reverses it back into place (so $\rho_{1} = \rho_{2}$).
    Therefore, the labels that are not to be permuted end up exactly where they started once the algorithm is complete.
   
    Now we analyze the routing time.
    Let $d_i = x_{i+1} - x_i - 1$ for $i \in [k-1]$.
    As in the algorithm, let $t$ be the largest index for which $x_t \leq \floor{n/2}$.
    Then, for $1 \leq i \leq t - 1$, we have $\length{\rho_i} = d_i + i$, and, for $t+2 \leq j \leq k$, we have $\length{\rho_j} = d_{j-1} + k - j$.
    Moreover, we have $\length{\rho_t} = \floor{n / 2} - x_t - 1 + t$ and $\length{\rho_{t+1}} = x_{t+1} - \floor{n / 2} + k - t$.
    From all reversals in the first part of \cref{alg:middle-exchange}, $\Bar \rho$, 
    consider those that are performed on the left side of the median (position $\lfloor n/2 \rfloor$ of the path).
    The routing time of these reversals is
    \begin{align}
    \begin{split}
        \frac{1}{3} \sum_{i = 1}^t \length{\rho_i}+1 
        &= \frac{1}{3} \left(\floor{n / 2} - x_t - 1 \right) + \frac{1}{3} \sum_{i = 1}^t \left( d_i + i +1 \right) \\
        &= \frac{t(t+1)}{6} + \frac{1}{3}\left(\floor{n/2} - x_t - 1 \right) +  \sum_{i=1}^t \left( x_{i+1} - x_i \right) \\
        &= \bigo{t^2} + \frac{1}{3} \left(\floor{n/2} - x_1\right) \\
        &\leq \frac{n}{6} + \bigo{k^2}.
    \end{split}
    \end{align}
    By a symmetric argument, the same bound holds for the compression step on the right half of the median. 
    Because both sides can be performed in parallel, the total cost for the compression step is at most $n/6 + O(k^{2})$.
    The inner permutation step can be done in time at most $k$ using odd-even sort.
    The cost to perform the dilation step is also at most $n/6 + O(k^{2})$ because dilation is the inverse of compression. 
    Thus, the total routing time for \cref{alg:middle-exchange} is at most $2(n/6 + O(k^{2})) + k = n/3 + O(k^{2})$.
\end{proof}

% This result can be generalized to any graph $G$ with $k  =  o(\sqrt{r(G)})$ qubits to be moved.

It follows that sparse permutations on the path with $k=o(\sqrt n)$ can be implemented using reversals with a full asymptotic factor of $3$ speedup.

\subsection{General graphs}

We now present a more general result for implementing sparse permutations on an arbitrary graph.

\begin{theorem} \label{thm:sparse-lower-bound-general-graph}
    Let $G = (V,E)$ be a graph with radius $r$ and $\pi$ a permutation of vertices.
    Let $S = \set{v \in V : \pi_v \not= v}$.
    Then $\pi$ can be routed in time at most $2r/3 + \bigo{|S|^2}$.
\end{theorem}

\begin{proof}
We route $\pi$ using a procedure similar to \cref{alg:middle-exchange}, consisting of the same three steps adapted to work on a spanning tree of $G$: compression, inner permutation, and dilation.
Dilation is the inverse of compression and the inner permutation step can be performed on a subtree
consisting of just $k=|S|$ nodes by using the \RoutingViaMatchings{} algorithm for trees in $3k/2 + \bigo{\log k}$ time ~\cite{Zhang1999}.
It remains to show that compression can be performed in $r/3 + \bigo{k^2}$ time.

We construct a \emph{token tree} $\mathcal T$ that reduces the compression step to routing on a tree.
Let $c$ be a vertex in the \emph{center} of $G$, i.e., a vertex with distance at most $r$ to all vertices.
Construct a shortest-path tree $\mathcal T'$ of $G$ rooted at $c$, say, using breadth-first search.
We assign a token to each vertex in $S$.
Now $\mathcal T$ is the subtree of $\mathcal T'$ formed by removing all vertices $v \in V(\mathcal T')$ for which the subtree rooted at $v$ does not contain any tokens, as depicted in \cref{fig:shortestpathtree}.
In $\mathcal T$, call the first common vertex between paths to $c$ from two distinct tokens an \emph{intersection} vertex,
and let $\mathcal I$ be the set of all intersection vertices.
Note that if a token $t_1$ lies on the path from another token $t_2$ to $c$,
then the vertex on which $t_1$ lies is also an intersection vertex.
Since $\mathcal T$ has at most $k$ leaves, $\abs{\mathcal I} \le k-1$.

% Let $c$ be a vertex in the \emph{center} of $G$, i.e., a vertex with distance at most $r$ to all vertices.
% Using, e.g., breadth-first search, construct a \emph{shortest-path tree} $\mathcal T'$ for $G$ rooted at $c$,
% that is, a spanning tree of $G$ such that the path from each vertex $v$ to $c$ on $\mathcal T'$ is a shortest path.
% Since we are only interested in the vertices in $Q$, consider the subgraph $\mathcal T$ of $\mathcal T'$ with all subtrees that do not contain vertices in $Q$ removed (\cref{fig:shortestpathtree}). 

% There are at most $k-1$ vertices in $\mathcal T$ with degree strictly greater than 2 that we call \emph{intersections}.
% Let $v_1, \dots, v_j$, $j \leq k - 1$, be the intersections ordered such that $d(v_1, c) \ge \dots \ge d(v_j, c)$,
% for distance function $d$.

\begin{figure}
    \centering
    \input{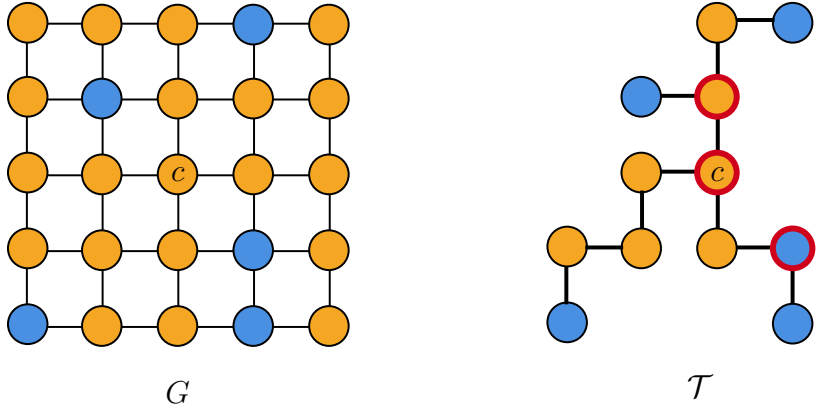}
    \caption{Illustration of the token tree $\mathcal T$ in \cref{thm:sparse-lower-bound-general-graph} for a case where $G$ is the $5 \times 5$ grid graph. Blue circles represent vertices in $S$ and orange circles represent vertices not in $S$. Vertex $c$ denotes the center of $G$. Red-outlined circles represent intersection vertices. In particular, note that one of the blue vertices is an intersection because it is the first common vertex on the path to $c$ of two distinct blue vertices.}
    \label{fig:shortestpathtree}
\end{figure}

For any vertex $v$ in $\mathcal T$, let the \emph{descendants} of $v$ be the vertices $u \not= v$ in $\mathcal T$ whose path on $\mathcal T$ to $c$ includes $v$.
Now let $\mathcal T_v$ be the subtree of $\mathcal T$ rooted at $v$, i.e., the tree composed of $v$ and all of the descendants of $v$.
We say that all tokens have been \emph{moved up to} a vertex $v$ if for all vertices $u$ in $\mathcal T_v$ without a token, $\mathcal \mathcal T_u$ also does not contain a token.
The compression step can then be described as moving tokens up to $c$.

% TBS
    \begin{algorithm}[tbp]
        \Input{
            A vertex $v$ in token tree $\mathcal T$
        }
        \SetKwFunction{FTripSort}{MoveUpTo}
        \Fn{\FTripSort{$v$}}{%
            \If(\tcp*[f]{Base case}\label{btc:base-case}){$\mathcal T_v$ contains no intersection vertices other than $v$}{
                \For{each leaf node $u \in V(\mathcal T_v)$}{
                    \If{$u$ is the first leaf node}{
                        Perform reversal from $u$ to $v$.
                    }
                    \Else{
                        Perform reversal from $u$ to $v$, exclusive.
                    }
                }
                \Return{}
            }
            
            \For{each descendant $b$ of $v$}{ \label{btc:forloop}
                $w := $ the intersection vertex in $\mathcal T_b$ closest to $b$ \tcp*{may include $b$} \label{btc:w}
                $\texttt{MoveUpTo}(w)$\; \label{btc:recursive-call}
                $m := $ the number of tokens from $S$ in $\mathcal T_b$\;
                $l(p) := $ the length of the path $p$ from $w$ to $b$ in $\mathcal T_v$\;
                \If(\tcp*[f]{Enough room on $p$, form a path of tokens at $b$}){$l(p) \geq m$}{\label{btc:case1}
                    Route the $m$ tokens in $\mathcal T_b$ to the first $m$ vertices of $p$ using \RoutingViaMatchings{}.\; \label{btc:routingviamatchings1}
                    Perform a reversal on the segment starting at $w$ and ending at $b$. \label{btc:moveuptob}
                }
                \Else(\tcp*[f]{Not enough room on $p$, form a tree of tokens rooted at $b$}){
                    Route the $m$ tokens in $\mathcal T_b$ as close as possible to $b$ using \RoutingViaMatchings{}.\; \label{btc:routingviamatchings2}
                }\label{btc:case2}
            } 
            \If(\tcp*[f]{Put token on root $v$}){$v$ has no token}{ \label{btc:step1}
                Perform a reversal on the segment starting from $v$ and ending at a vertex $u$ in $\mathcal T_v$ with a token such that no descendant of $u$ has a token.\label{btc:finalrev}
            }
        }
       \caption{An algorithm that recursively moves all tokens from $S$ that lie on $\mathcal T_v$ up to an intersection vertex $v$.}\label{alg:move-up-to}
    \end{algorithm}

\begin{figure}
\begin{tabular}{c|c}
    \begin{subfigure}[b]{0.49\textwidth}
        
        \begin{center}
        \resizebox{0.8\columnwidth}{!}
            {
            \input{plots/case1}  }  
        \end{center}
        \caption{When $l(p) \geq m$. In this case, $l(p) = 7 \geq 5 = m$.}\label{fig:btc-case1}
    \end{subfigure}
    &
    \begin{subfigure}[b]{0.49\textwidth}
        \begin{center}
        \resizebox{0.7\columnwidth}{!}
            {
            \input{plots/case2}}
        \end{center}
        \caption{When $l(p) < m$. In this case, $l(p) = 3 < 5 = m$.}\label{fig:btc-case2}
    \end{subfigure}
\end{tabular}
\caption{%
    An example of moving the $m$ tokens in $\mathcal T_w$ up to $b$ (\cref{btc:case1}--\cref{btc:case2} in \cref{alg:move-up-to}).
}\label{fig:btc}
\end{figure}

We describe a recursive algorithm for doing so in \cref{alg:move-up-to}.
The base case considers the trivial case of a subtree with only one token.
Otherwise, we move all tokens on the subtrees of descendant $b$ up to the closest intersection $w$ using recursive calls
as illustrated in \cref{fig:btc}.
Afterwards, we need to consider whether the path $p$ between $v$ and $w$ has enough room to store all tokens.
If it does, we use a \RoutingViaMatchings{} algorithm for trees to route tokens from $w$ onto $p$,
followed by a reversal to move these tokens up to $v$.
Otherwise, the path is short enough to move all tokens up to $v$ by the same \RoutingViaMatchings{} algorithm.

We now bound the routing time on $\mathcal T_{w_1}$ of \FTripSort{$w_1$}, for any vertex $w_1 \in V(\mathcal T)$.
First note that all operations on subtrees $\mathcal T_b$ of $\mathcal T_{w_1}$ are independent and can be performed in parallel.
Let $w_1, w_2, \dots, w_t$ be the sequence of intersection vertices that \FTripSort{$\cdot$} is recursively called on that dominates the routing time of \FTripSort{$w_1$}.
Let $d_w$, for $w \in V(\mathcal T_{w_1})$, be the distance of $w$ to the furthest leaf node in $\mathcal T_w$.
Assuming that the base case on \cref{btc:base-case} has not been reached,
we have a routing time of
\begin{equation}
    T(w_1) \leq T(w_2) + \frac{d_{w_1}-d_{w_2}}{3} + \bigo{k},
\end{equation}
where $\bigo{k}$ bounds the time required to route $m \le k$ tokens on a tree of size at most $2m$ following the recursive \FTripSort{$w_2$} call~\cite{Zhang1999}.
We expand the time cost $T(w_i)$ of recursive calls until we reach the base case of $w_t$ to obtain
\begin{equation}
    T(v) \le T(w_t) + \sum_{i=1}^{t-1}\left( \frac{d_{w_i} - d_{w_{i+1}}}{3} + \bigo{k}\right) = T(w_t) + \frac{d_{w_1} - d_{w_t}}{3} + t\cdot \bigo{k} \le \frac{d_{w_1}}{3} + (t+1)\bigo{k}.
\end{equation}
Since $d_v \le r$ and $t \le k$, this shows that compression can be performed in $r/3 + \bigo{k^2}$ time.
\end{proof}

In general, a graph with radius $r$ and diameter $d$ will have $d/2 \leq r \leq d$.
Using \cref{thm:sparse-lower-bound-general-graph}, this implies that for a graph $G$ and a sparse permutation with $k = o(\sqrt{r})$, the bound for the routing time will be between $d/3 + o(d)$ and $2d/3 + o(d)$.
Thus, for such sparse permutations, using reversals will always asymptotically give us a constant-factor worst-case speedup over any \swap-only protocol since $\rt(G) \geq d$.
Furthermore, for graphs with $r = d/2$, we can asymptotically achieve the full factor of 3 speedup.

%%%%%%%%%%%%%%%%%%%%%%%%%%%%%%%%%%%%%%%%%%%%%%%%%%%%%%%%%%%%%%%%%%%%%%%%%%%%%%%%
\section{Algorithms for routing on the path} \label{results}

% Generic Divide Conquer 
\begin{algorithm}[tbp] % I prefer top/bottom/float page. Be cautious with "[H]ere".
    % Define constraints on the input. What it is.
    \Input{%
        $\pi$, a permutation of a contiguous subset of $[n]$.
    }
    % Define functions. The first argument becomes a macro (\FDivConquer in this case)
    \SetKwFunction{FDivConquer}{GenericDivideConquer}
    % Define the function
    \Fn{\FDivConquer{\begin{tt}BinarySorter\end{tt}, $\pi$}}{
        \If{$\length{\pi}$ = 1}{%
            \Return $\emptyset$\;
        }
        $B := \texttt{BinaryLabeling}(\pi)$\;
        $\overline{\rho} := \texttt{BinarySorter}(B)$\;
        $\pi := \pi \cdot \overline{\rho}$\;
        % Can use defined functions (SetKwFunction) as function calls
        % Note that we have to exit math mode :(
        $\overline{\rho} = \overline{\rho} \doubleplus \mbox{\FDivConquer{\text{\texttt{BinarySorter}}, $\pi[0,\floor*{\frac{n}{2}}]$}}$\;
        $\overline{\rho} = \overline{\rho} \doubleplus \mbox{\FDivConquer{\text{\texttt{BinarySorter}}, $\pi[\floor*{\frac{n}{2}}+1, |\pi|]$}}$\;
        \Return $\overline{\rho}$
    }
    \caption{Divide-and-conquer algorithm for recursively sorting $\pi$. $\texttt{BinaryLabeling}(\pi)$ is a subroutine that uses \cref{eq:indicator} to transform $\pi$ into a bitstring, and $\texttt{BinarySorter}$ is a subroutine that takes as input the resulting binary string and returns an ordered reversal sequence $\Bar{\rho}$ that sorts it.}%
    \label{alg:divconqgeneric}
\end{algorithm}

% Concatenating two reversal series is written as $\Bar{\rho_{1}} \doubleplus \Bar{\rho_{2}}$. Applying $\rho$ or $\Bar{\rho}$ to a sequence $S$ is written as $S \cdot \rho$.

Our general approach to implementing permutations on the path relies on the divide-and-conquer strategy described in \cref{alg:divconqgeneric}. It uses a correspondence between implementing permutations and sorting binary strings, where the former can be performed at twice the cost of the latter. This approach is inspired by \cite{pinter-skiena-genomic} and \cite{bender-bounds} who use the same method for routing by reversals in the sequential case.

First, we introduce a binary labeling using the indicator function
\begin{equation}\label{eq:indicator}
	I(v) = \begin{cases*}
		0 & if $v < n/2$, \\
		1 & otherwise.
	\end{cases*}
\end{equation}
This function labels any permutation $\pi\in\mathcal{S}_n$ by a binary string $I(\pi) \coloneqq (I(\pi_1), I(\pi_2), \dots, I(\pi_n))$.  
Let $\pi$ be the target permutation, and $\sigma$ any permutation such that $I(\pi\sigma^{-1}) = (0^{\lfloor n/2\rfloor} 1^{\lceil n/2\rceil})$. Then it follows that $\sigma$ divides $\pi$ into permutations $\pi_L,\pi_R$ acting only on the left and right halves of the path, respectively, i.e., $\pi=\pi_L\cdot\pi_R\cdot\sigma$. We find and implement $\sigma$ via a binary sorting subroutine, thereby reducing the problem into two subproblems of length at most $\lceil n/2\rceil$ that can be solved in parallel on disjoint sections of the path. Proceeding by recursion until all subproblems are on sections of length at most 1, the only possible permutation is the identity and $\pi$ has been implemented. 
Because disjoint permutations are implemented in parallel, the total routing time is $T(\pi) = T(\sigma) + \max(T(\pi_L),T(\pi_R))$.

We illustrate \cref{alg:divconqgeneric} with an example, where the binary labels are indicated below the corresponding destination indices:
\begin{equation} \label{eq:divide-conquer-example}
    \begin{gathered}
    \texttt{7\;6\;0\;2\;5\;1\;3\;4\;} 
        \xrightarrow{\text{\footnotesize{label}}}
    \underwrite[0.5pt]{\texttt{7\;6\;0\;2\;5\;1\;3\;4\;}}{\texttt{1\;1\;0\;0\;1\;0\;0\;1\;}} 
        \xrightarrow{\text{\footnotesize{sort}}} 
    \underwrite[0.5pt]{\texttt{0\;3\;1\;2\;5\;7\;6\;4\;}}{\texttt{0\;0\;0\;0\;1\;1\;1\;1\;}} 
        \xrightarrow{\text{\footnotesize{label}}}
    \underwrite[0.5pt]{\texttt{0\;3\;1\;2} \quad \texttt{5\;7\;6\;4\;}}{\texttt{0\;1\;0\;1} \quad \texttt{0\;1\;1\;0\;}} \\
        \hspace{10cm} \downarrow \text{\footnotesize{ sort }} \\
    \underwrite[0.5pt]{\texttt{0\;1} \quad \texttt{2\;3} \quad \texttt{4\;5} \quad \texttt{6\;7\;}}{\texttt{0\;1} \quad \texttt{0\;1} \quad \texttt{0\;1} \quad \texttt{0\;1\;}}
        \xleftarrow{\text{\footnotesize{sort}}}
    \underwrite[0.5pt]{\texttt{0\;1} \quad \texttt{3\;2} \quad \texttt{5\;4} \quad \texttt{6\;7\;}}{\texttt{0\;1} \quad \texttt{1\;0} \quad \texttt{1\;0} \quad \texttt{0\;1\;}}
        \xleftarrow{\text{\footnotesize{label}}}
    \underwrite[0.5pt]{\texttt{0\;1\;3\;2} \quad \texttt{5\;4\;6\;7\;}}{\texttt{0\;0\;1\;1} \quad \texttt{0\;0\;1\;1\;}}
    \end{gathered}
\end{equation}
Each labeling and sorting step corresponds to an application of \cref{eq:indicator} and \texttt{BinarySorter}, respectively, to each subproblem. Specifically, in \cref{eq:divide-conquer-example}, we use TBS (\cref{alg:TBS}) to sort binary strings.

% TBS
    \begin{algorithm}[tbp]
        \Input{
            $B$, a binary string 
        }
        \SetKwFunction{FTripSort}{TripartiteBinarySort}
        \Fn{\FTripSort{$B$}}{%
            \If{$\length{B} = 1$}{%
                \Return $\emptyset$
            }
            $m_{1} := \left\lfloor\frac{|B|}{3}\right\rfloor$\;
            $m_{2} := \left\lfloor\frac{2|B|}{3}\right\rfloor$\;
            $\overline{\rho} := \texttt{TripartiteBinarySort}(B[0, m_{1}])$\;
            $\overline{\rho} := \overline{\rho} \doubleplus
            \texttt{TripartiteBinarySort}(B[m_{1} + 1,m_2] \oplus \texttt{11}\dots\texttt{1})$\;
            \label{line:backwardsort}
                % \tcp*{$\oplus$ denotes bitwise XOR, so we sort the  middle third backwards}
            $\overline{\rho} := \overline{\rho} \doubleplus \texttt{TripartiteBinarySort}(B[m_{2} +1, |B|])$\;
            $B \gets$ apply reversals in $\bar \rho$ to $B$\;
            $i := $ index of first \texttt{1} in $B$\;
            $j := $ index of last \texttt{0} in $B$\;
            \Return $\overline{\rho} \doubleplus \rho(i,j)$\;
        }
       \caption{Tripartite Binary Sort (TBS). We let $\rho(i,j)$ denote a reversal on the subsequence $S[i,j]$ (inclusive of $i$ and $j$). In line \ref{line:backwardsort}, $\oplus \texttt{11}\dots\texttt{1}$ indicates that we flip all the bits, so that we sort the middle third backwards.}\label{alg:TBS}
    \end{algorithm}
\newcommand{\Alg}[1]{\texttt{GDC(#1)}}

We present two algorithms for \texttt{BinarySorter},
which perform the work in our sorting algorithm.
The first of these binary sorting subroutines is Tripartite Binary Sort (TBS, \cref{alg:TBS}).
TBS works by splitting the binary string into nearly equal (contiguous) thirds, recursively sorting these thirds, and merging the three sorted thirds into one sorted sequence.
We sort the outer thirds forwards and the middle third backwards which allows us to merge the three segments using at most one reversal.
For example, we can sort a binary string as follows:
\begin{equation}
    \begin{gathered}
    \texttt{010011100011010011110111001} \\
    \texttt{010011100~~011010011~~110111001} \\
    \text{\small{TBS}} \downarrow \hspace{0.9cm} \text{\small{TBS}} \downarrow  \text{\small{backwards}} \hspace{0.9cm} \downarrow \text{\small{TBS}} \\
    \texttt{000001111~~111110000~~000111111} \\
    \texttt{00000}\underbracket{\texttt{1111111110000000}}\texttt{11111} \\
    \texttt{00000000000011111111111111},
    \end{gathered}\label{eq:binarySortingExample}
\end{equation}
where the arrows with TBS indicate recursive calls to TBS and the bracket indicates the reversal to merge the segments.
Let \Alg{TBS} denote \cref{alg:divconqgeneric} when using TBS to sort binary strings, where \texttt{GDC} stands for \texttt{GenericDivideConquer}.

The second algorithm is an adaptive version of TBS (\cref{alg:AdaptiveTBSexhaustive})
that, instead of using equal thirds, adaptively chooses the segments' length.
Adaptive TBS considers every pair of partition points, $0 \leq i \leq j < n-1$, that would split the binary sequence into two or three sections: $B[0,i]$, $B[i+1, j]$, and $B[j+1, n-1]$ (where $i = j$ corresponds to no middle section).
For each pair, it calculates the minimum cost to recursively sort the sequence using these partition points.
Since each section can be sorted in parallel, the total \emph{sorting time} depends on the maximum time needed to sort one of the three sections and the cost of the final merging reversal.
Let \Alg{ATBS} denote \cref{alg:divconqgeneric} when using Adaptive TBS to sort binary strings.

Notice that the partition points selected by TBS are considered by the Adaptive TBS algorithm and are selected by Adaptive TBS only if no other pair of partition points yields a faster sorting time.
Thus, for any permutation, the sequence of reversals found by Adaptive TBS costs no more than that found by TBS.
However, TBS is simpler to implement and will be faster than Adaptive TBS in finding the sorting sequence of reversals.

% Adaptive TBS exhaustive search
\begin{algorithm}[tbp]
        \Input{%
            $B$, a binary string 
        }
        \textbf{function} \texttt{AdaptiveTripartiteBinarySort}$(B)$: \\
        $\overline{\rho} := \emptyset$\;
        \For{$i = 0$ to $n-2$}{%
            \For{$j = i$ to $n-2$}{%
                $\overline{\rho_0} = \texttt{AdaptiveTripartiteBinarySort}(B[0,i])$\;
                $c_0 := cost(\overline{\rho_0})$\;
                $\overline{\rho_1} = \texttt{AdaptiveTripartiteBinarySort}(B[i+1,j])$\;
                $c_1 := cost(\overline{\rho_1})$\;
                $\overline{\rho_2} = \texttt{AdaptiveTripartiteBinarySort}(B[j+1,n-1])$\;
                $c_2 := cost(\overline{\rho_2})$\;
                $r := \text{ cost of merging reversal using }i\text{ and }j\text{ as partition points}$\;  %(N_1(B[0,i]) + j-i + N_0(B[j+1,n-1]) + 1) / 3
                \If{$\overline{\rho} = \emptyset$ or $\max\{c_0, c_1, c_2\} + r < cost(\overline{\rho})$}{%
                    $\overline{\rho} := \overline{\rho_0} \doubleplus \overline{\rho_1} \doubleplus \overline{\rho_2}$\;
                }
            }
        }
        \Return $\overline{\rho}$
    \caption{Adaptive TBS.
        For the sake of clarity, we implement an exhaustive search over all possible ways to choose the partition points.
        However, we note that the optimal partition points can be found in polynomial time by using a dynamic programming method~\cite{bender-bounds}. 
    }
    \label{alg:AdaptiveTBSexhaustive}
\end{algorithm}

\subsection{Worst-case bounds} \label{worstcase}
In this section, we prove that all permutations of sufficiently large length $n$ can be sorted in time strictly less than $n$
using reversals.
Let $n_x(b)$ denote the number of times character $x \in \{0,1\}$ appears in a binary string $b$, and let $T(b)$ (resp., $T(\pi)$) denote the best possible sorting time to sort $b$ (resp., implement $\pi$) with reversals.
Assume all logarithms are base 2 unless specified otherwise.

\begin{lemma} \label{TBSboundlemma}
    Let $b \in \{0,1\}^{n}$ such that $n_x(b) < c n + O(\log n)$, where $c \in [0,1/3]$ and $x \in \{0,1\}$. Then, $T(b) \leq \left(c / 3 + 7/18 \right)n + \bigo{\log n}$.
\end{lemma}
\begin{proof}
    To achieve this upper bound, we use TBS (\cref{alg:TBS}). 
    There are $\lfloor\log_3 n\rfloor$ steps in the recursion, which we index by $j\in\{0,1,\ldots, \lfloor\log_3 n\rfloor\}$, with step 0 corresponding to the final merging step.
    Let $\length{\rho_{j}}$ denote the size of the longest reversal in recursive step $j$ that merges the three sorted subsequences of size $n/3^{j+1}$.
    The size of the final merging reversal $\rho_{0}$ can be bounded above by $(c + 2/3)n + O(\log n)$ because $\length{\rho_{0}}$ is maximized when every $x$ is contained in the leftmost third if $x = 1$ or the rightmost third if $x = 0$.
    So we have
    \begin{align}
        T(b) &\leq \left(\sum_{j=0}^{\log_{3}{n}}\frac{\length{\rho_{j}}}{3}\right) + \bigo{\log n}
        \leq \left(\frac{c}{3} + \frac{2}{9} \right)n +O(\log n) + \left(\sum_{j = 1}^{\log_{3}{n}} \frac{\abs{\rho_j}}{3}\right) + \bigo{\log n} \\
        &\leq \left( \frac{c}{3} + \frac{7}{18}\right) n + \bigo{\log n},
    \end{align}
    where we used $\abs{\rho_j} \leq n/3^j$ for $j\geq 1$.
\end{proof}

Now we can prove a bound on the cost of a sorting series found by Adaptive TBS for any binary string of length $n$.

\begin{theorem} \label{bitstringbound}
    For all bit strings $b \in \{0, 1\}^n$ of arbitrary length $n\in\mathbb{N}$,
    $T(b) \leq \left(1 / 2 - \varepsilon \right) n + \bigo{\log n} \approx 0.483n + \bigo{\log n}$,
    where $\varepsilon = 1/3 - 1/\sqrt{10}$.
\end{theorem}
\begin{proof}
    Let $b \in \{0, 1\}^n$ for some $n \in \N$.
    Partition $b$ into three sections $b = b_{1}b_{2}b_{3}$ such that $\length{b_{1}} = \length{b_{3}} = \floor{n/3}$ and $\length{b_{2}} = n - 2\floor{n/3}$.
    Since $\floor{n/3} = n/3 - d$ where $d \in \{0,1/3,2/3\}$, we write $\length{b_{1}} = \length{b_{2}} = \length{b_{3}} = n/3 + O(1)$ for the purposes of this proof.
    Recall that if segments $b_1$ and $b_3$ are sorted forwards and segment $b_2$ is sorted backwards, the resulting segment can be sorted using a single reversal, $\rho$ (see the example in \cref{eq:binarySortingExample}).
    Then we have
    \begin{equation}\label{TBSgenericbound}
        T(b) \leq \max(T(b_1), T'(b_2), T(b_3)) + \frac{\length{\rho} + 1}{3},
    \end{equation}
    where $T'(b_2)$ is the time to sort $b_2$ backwards using reversals. 
    
    We proceed by induction on $n$.
    For the base case, it suffices to note that every binary string can be sorted using reversals and, for finitely many values of $n \in \N$, any time needed to sort a binary string of length $n$ exceeding $\left(1/2 - \varepsilon \right) n$ can be absorbed into the $\bigo{\log n}$ term.
    Now assume $T(b) \leq \left(1/2 - \varepsilon\right)k + \bigo{\log k}$ for all $k < n$, $b \in \{0,1\}^k$.
    
    \textbf{Case 1: $n_0(b_1) \geq 2 \varepsilon n$ or $n_1(b_3) \geq 2 \varepsilon n$.}
    In this case, $\length{\rho} \leq n - 2\varepsilon n$, so
    \begin{equation}
        T(b) \leq \frac{n - 2\varepsilon n + 1}{3} + \max(T(b_1), T'(b_2), T(b_3)) \leq \left(\frac{1}{2} - \varepsilon\right)n + \bigo{\log n}
    \end{equation}
    by the induction hypothesis.
    
    \begin{figure}
        \centering
        \tikzset{every picture/.style={line width=0.75pt}} %set default line width to 0.75pt        

\begin{tikzpicture}[x=0.75pt,y=0.75pt,yscale=-1,xscale=1]
%uncomment if require: \path (0,124); %set diagram left start at 0, and has height of 124

%Shape: Rectangle [id:dp42983375176056526] 
\draw   (43,82) -- (241.67,82) -- (241.67,73) -- (43,73) -- cycle ;
%Shape: Rectangle [id:dp003631911449202274] 
\draw   (241.67,82) -- (440.33,82) -- (440.33,73) -- (241.67,73) -- cycle ;
%Shape: Rectangle [id:dp9919754824730549] 
\draw   (440.33,82) -- (639,82) -- (639,73) -- (440.33,73) -- cycle ;
%Shape: Rectangle [id:dp5017203070084941] 
\draw  [color={rgb, 255:red, 255; green, 0; blue, 0 }  ,draw opacity=1 ] (395,73) -- (396,73) -- (396,80) -- (395,80) -- cycle ;
%Shape: Rectangle [id:dp09825623350237622] 
\draw  [color={rgb, 255:red, 255; green, 0; blue, 0 }  ,draw opacity=1 ] (286,74) -- (287,74) -- (287,81) -- (286,81) -- cycle ;
%Curve Lines [id:da6225946554998514] 
\draw [color={rgb, 255:red, 255; green, 0; blue, 0 }  ,draw opacity=1 ]   (244,60) .. controls (262.1,45.52) and (276.92,55.61) .. (284.02,62.85) ;
\draw [shift={(286,65)}, rotate = 229.4] [fill={rgb, 255:red, 255; green, 0; blue, 0 }  ,fill opacity=1 ][line width=0.08]  [draw opacity=0] (8.93,-4.29) -- (0,0) -- (8.93,4.29) -- cycle    ;
%Curve Lines [id:da5358776908129448] 
\draw [color={rgb, 255:red, 255; green, 0; blue, 0 }  ,draw opacity=1 ]   (401.68,58.02) .. controls (420.42,45.25) and (435.27,58.32) .. (441,65) ;
\draw [shift={(399,60)}, rotate = 321.34] [fill={rgb, 255:red, 255; green, 0; blue, 0 }  ,fill opacity=1 ][line width=0.08]  [draw opacity=0] (8.93,-4.29) -- (0,0) -- (8.93,4.29) -- cycle    ;

% Text Node
\draw (225,87) node [anchor=north west][inner sep=0.75pt]  [font=\footnotesize] [align=left] {$\displaystyle n/3$};
% Text Node
\draw (423,85) node [anchor=north west][inner sep=0.75pt]  [font=\footnotesize] [align=left] {$\displaystyle 2n/3$};
% Text Node
\draw (244,22) node [anchor=north west][inner sep=0.75pt]  [font=\scriptsize] [align=left] {$\displaystyle \textcolor[rgb]{1,0,0}{\frac{2\varepsilon n}{3-6\varepsilon }}$};
% Text Node
\draw (409,21) node [anchor=north west][inner sep=0.75pt]  [font=\scriptsize] [align=left] {$\displaystyle \textcolor[rgb]{1,0,0}{\frac{2\varepsilon n}{3-6\varepsilon }}$};
% Text Node
\draw (87,94) node [anchor=north west][inner sep=0.75pt]  [font=\small] [align=left] {$\displaystyle n_{0}( b_{1}) \ < \ 2\varepsilon n$};
% Text Node
\draw (515,94) node [anchor=north west][inner sep=0.75pt]  [font=\small] [align=left] {$\displaystyle n_{1}( b_{3}) \ < \ 2\varepsilon n$};

\end{tikzpicture}
        \caption{Case 2 of \cref{bitstringbound}.  If there are few zeros and ones in the leftmost and rightmost thirds, respectively, we can shorten the middle section so that it can be sorted quickly.  Then, because each of the outer thirds contain far more zeros than ones (or vice versa), they can both can be sorted quickly as well.}
        \label{fig:adjusting-partition}
    \end{figure}
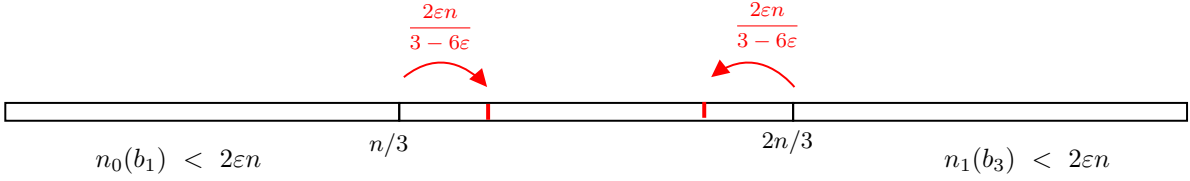
    
    \textbf{Case 2: $n_0(b_1) < 2\varepsilon n$ and $n_1(b_3) < 2 \varepsilon n$.}
    In this case, adjust the partition such that $\length{b_{1}} = \length{b_{3}} = n/3 + 2\varepsilon n/(3-6\varepsilon) - O(1)$ and consequently $\length{b_{2}} = n/3 - 4\varepsilon n / (3-6\varepsilon) + O(1)$, as depicted in \cref{fig:adjusting-partition}. 
    In this adjustment, at most $2\varepsilon n/(3-6\varepsilon)$ zeros are added to the segment $b_1$ and likewise with ones to $b_3$.
    Thus, $n_{1}(b_3) \leq 2\varepsilon n + 2\varepsilon n/(3-6\varepsilon) = \left( 1 + 1/(3-6\varepsilon)\right)2\varepsilon n$.
    Since $n = (3-6\varepsilon) \length{b_1} - O(1)$, we have
    \begin{equation}
        n_{1}(b_3) \leq \left( 1 + \frac{1}{3-6\varepsilon}\right)2\varepsilon ((3-6\varepsilon)\length{b_1} - O(1)) = (2-3\varepsilon)4\varepsilon \length{b_1} - O(1) .
    \end{equation}
    Let $c = (2-3\e)4\e = 2/15$. 
    Applying \cref{TBSboundlemma} with this value of $c$ yields
    \begin{align}
        T(b_{3}) &\leq \left( \frac{2}{45} + \frac{7}{18}\right)\length{b_1} + \bigo{\log \left(\length{b_1}\right)}
        = \left(\frac{1}{\sqrt{10}} - \frac{1}{6} \right) n + \bigo{\log n} .
    \end{align}
    Since $\length{b_1} = \length{b_3}$, we obtain the same bound $T(b_1) \leq (1/\sqrt{10} - 1/6)n + \bigo{\log n}$
    by applying \cref{TBSboundlemma} with the same value of $c$.
    
    By the inductive hypothesis, $T'(b_{2})$ can be bounded above by 
    \begin{align}
        T'(b_{2}) &\leq \left(\frac{1}{2} - \e\right)\left(\frac{n}{3} - \frac{4\varepsilon}{3 - 6 \varepsilon}n + O(1)\right) + \bigo{\log n}
        % &= \left(\frac{1}{6} - \e\right)n + \bigo{\log n}. \\
        = \left(\frac{1}{\sqrt{10}} - \frac{1}{6}\right)n + \bigo{\log n} .
    \end{align}
    Using \cref{TBSgenericbound} and the fact that $\length{\rho} \leq n$, we get the bound
    \begin{equation*}
        T(b) \leq \left(\frac{1}{\sqrt{10}} - \frac{1}{6} \right) n + \bigo{\log n} + \frac{n + 1}{3} = \left( \frac{1}{2} - \e \right)n + \bigo{\log n}
    \end{equation*}
    as claimed.
\end{proof}

This bound on the cost of a sorting series found by Adaptive TBS for binary sequences can easily be extended to a bound on the minimum sorting sequence for any permutation of length $n$.

\begin{corollary} \label{permutationbound}
    For a length-$n$ permutation $\pi$, $T(\pi) \leq \bigl(1/3 + \sqrt{2/5}\bigr) n + \bigo{\log^2 n} \approx 0.9658n + \bigo{\log^2 n}$.
\end{corollary}
\begin{proof}
    To sort $\pi$, we turn it into a binary string $b$ using \cref{eq:indicator}.
    Then let $\rho_1,\rho_2, \dots, \rho_m$ be a sequence of reversals to sort $b$.
    If we apply the sequence to get $\pi'=\pi\rho_1\rho_2\cdots\rho_m$, every element of $\pi'$ will be on the same half as its destination.
    We can then recursively perform the same procedure on each half of $\pi'$, continuing down until every pair of elements has been sorted.
    
    This process requires $\lfloor\log n \rfloor$ steps, and at step $i$, there are $2^i$ binary strings of length $\frac{n}{2^i}$ being sorted in parallel.
    This gives us the following bound to implement $\pi$:
    \begin{equation}
        T(\pi) \leq \sum_{i=0}^{\log n}T(b_i) ,
    \end{equation}
    where $b_i \in \{0,1\}^{n/2^i}$.
    Applying the bound from \cref{bitstringbound}, we obtain
    \begin{equation*}
        T(\pi) \leq \sum_{i=0}^{\log n}T(b_i) \leq \sum_{i=0}^{\log n} \left(\left( \frac{1}{6} + \frac{1}{\sqrt{10}} \right)\frac{n}{2^i} + \bigo{\log (n/2^i)}\right) = \left(\frac{1}{3} + \sqrt{\frac{2}{5}} \right) n + \bigo{\log^2 n} . \qedhere
    \end{equation*}
\end{proof}

%%%%%%%%%%%%%%%%%%%%%%%%%%%%%%%%%%%%%%%%%%%%%%%%%%%%%%%%%%%%%%%%%%%%%%%%%%%%%%%%
\section{Average-case performance} \label{averagecase}

\begin{figure}
    \begin{subfigure}[t]{0.5\textwidth}
        \centering
        \includegraphics[width=\textwidth]{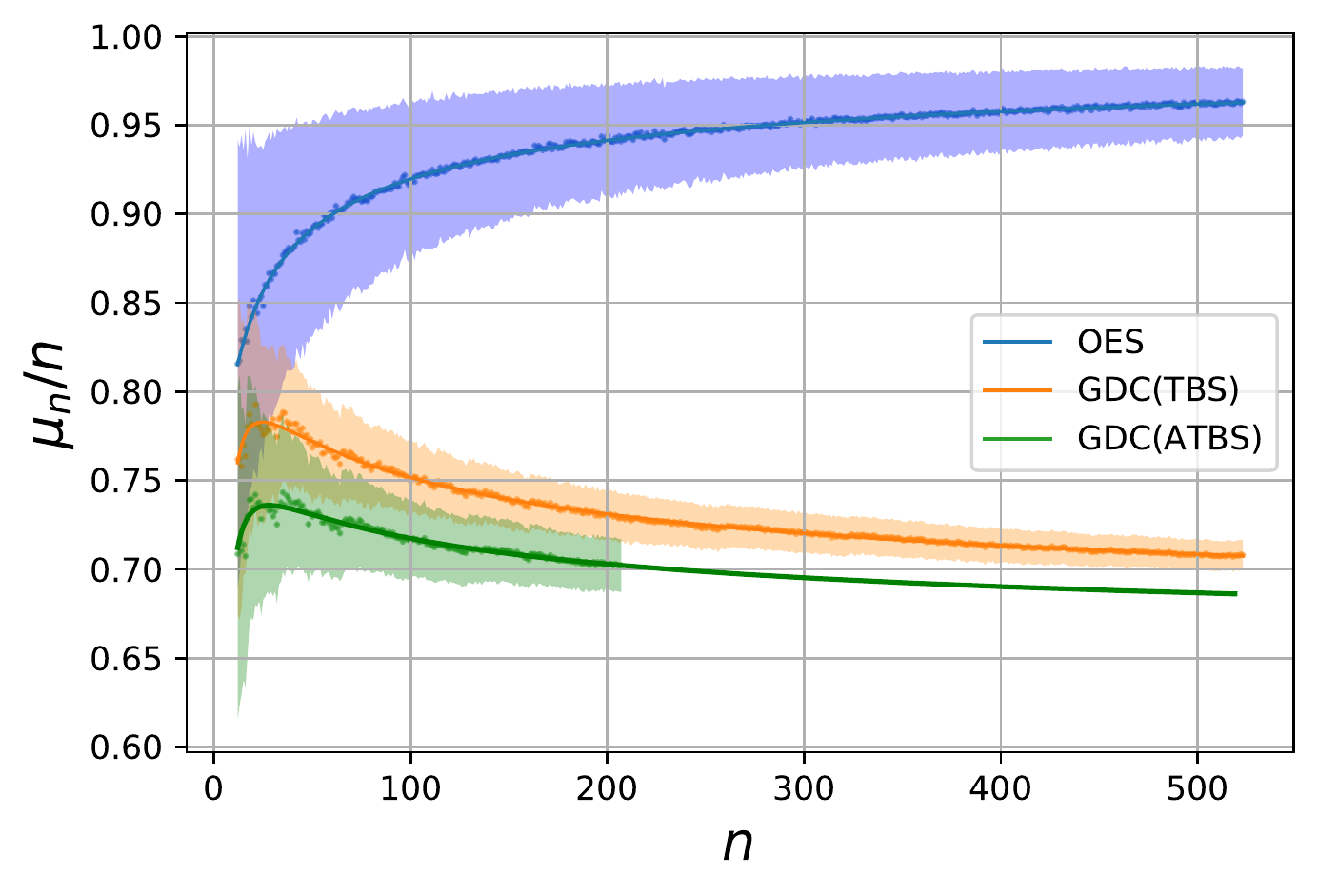}
        \caption{%
            Normalized mean routing time with std. deviation. 
        }
    \end{subfigure}
    \begin{subfigure}[t]{0.5\textwidth}
        \centering
        \includegraphics[width=\textwidth]{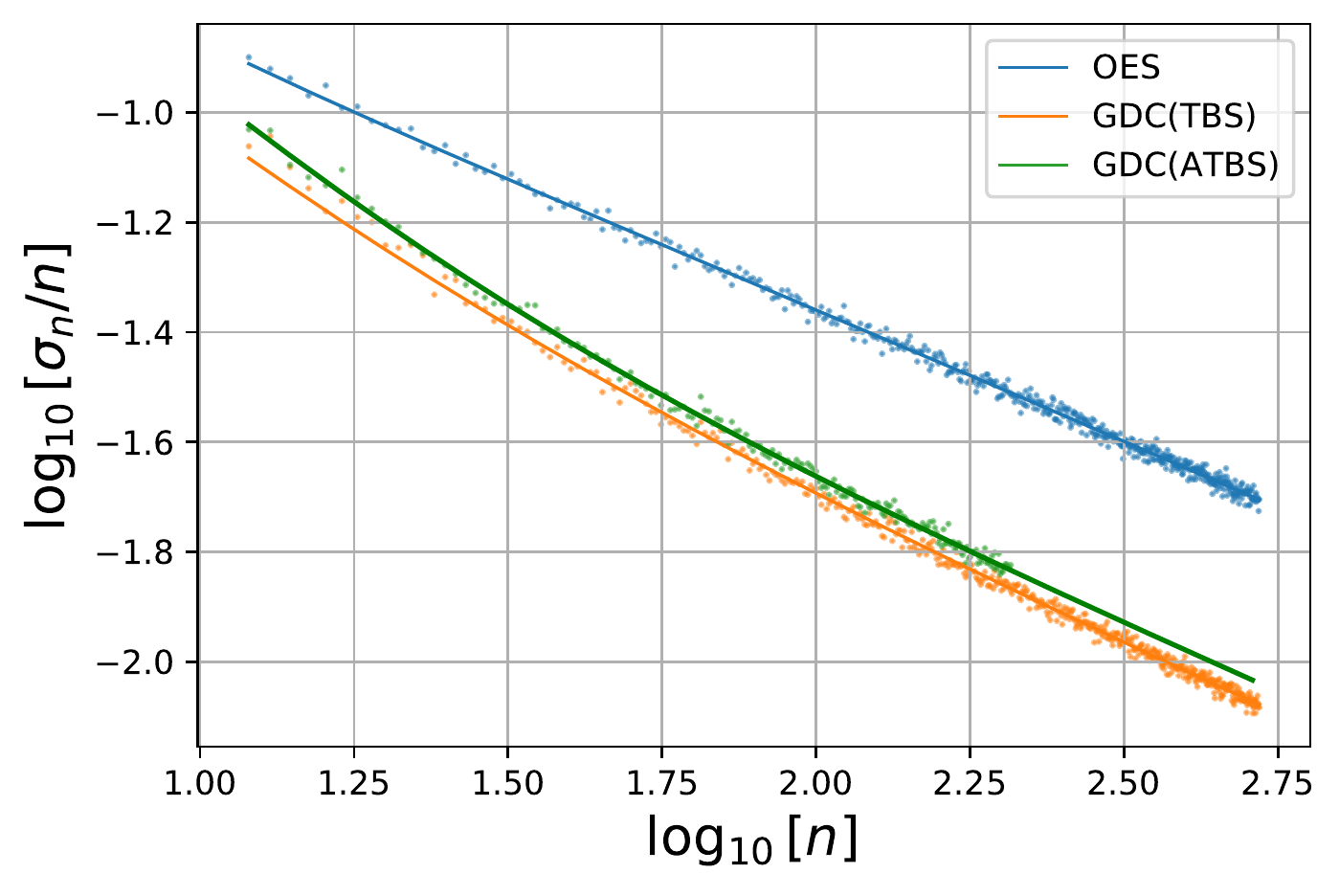}
        \caption{%
            Log normalized standard deviation of the routing time.
        }
    \end{subfigure}
    \caption{
    The mean routing time and fit of the mean routing time for odd-even sort (OES),
    and routing algorithms using Tripartite Binary Sort (\Alg{TBS}) and Adaptive TBS (\Alg{ATBS}).
    We exhaustively search for $n<12$ and sample 1000 permutations uniformly at random otherwise.
    We show data for \Alg{ATBS} only for $n \leq 207$ because it becomes too slow after that point. 
    We find that the fit function $\mu_n = an + b\sqrt{n} + c$ fits the data with an $R^2 > 99.99\%$ (all three algorithms).
    For OES, the fit gives $a \approx 0.9999$; for \Alg{TBS}, $a \approx 0.6599$; and for \Alg{ATBS}, $a \approx 0.6513$.
    Similarly, for the standard deviation, we find that the fit function $\sigma^2_n = an + b\sqrt{n} + c$ fits the data with $R^2 \approx 99\%$ (all three algorithms), suggesting that the normalized deviation of the performance about mean scales as $\sigma_n/n = \Theta(n^{-0.5})$ asymptotically.}
    \label{fig:asymp-mean}
\end{figure}

So far we have presented worst-case bounds that provide a theoretical guarantee on the speedup of quantum routing over classical routing. However, the bounds are not known to be tight, and may not accurately capture the performance of the algorithm in practice. 

In this section we show better performance for the 
\emph{average-case} routing time, the expected routing time of the algorithm on a permutation chosen uniformly at random from $\mathcal{S}_n$.
We present both theoretical and numerical results on the average routing time of swap-based routing (such as odd-even sort) and quantum routing using TBS and ATBS.
We show that on average, \Alg{TBS} (and \Alg{ATBS}, whose sorting time on any instance is at least as fast) beats swap-based routing by a constant factor $2/3$. We have the following two theorems, whose proofs can be found in \cref{OES-appendix,TBS-appendix}, respectively.

\begin{theorem}
\label{avg-OES}
The average routing time of any \swap-based procedure is lower bounded by $n-o(n)$.
\end{theorem}

\begin{theorem}\label{avg-TBS}
The average routing time of \Alg{TBS} is $2n/3 + O(n^{\alpha})$ for a constant $\alpha\in \bigl(\frac{1}{2},1\bigr)$.
\end{theorem}

These theorems provide average-case guarantees, yet do not give information about the non-asymptotic behavior. Therefore, we test our algorithms on random permutations for instances of intermediate size.

Our numerics~\cite{reversal-sort-code} show that \cref{alg:divconqgeneric} has an average routing time that is well-approximated by $c\cdot n + o(n)$, where $2/3 \lesssim c < 1$, using TBS or Adaptive TBS as the binary sorting subroutine, for permutations generated uniformly at random. 
Similarly, the performance of odd-even sort (OES) is well-approximated by $n + o(n)$. 
%We are interested in the relationship between the average routing time of the returned reversal sequences and the permutation length. 
Furthermore, the advantage of quantum routing is evident even for fairly short paths.
We demonstrate this by sampling 1000 permutations uniformly from $\mathcal{S}_{n}$ for $n \in [12,512]$, and running OES and \Alg{TBS} on each permutation.
Due to computational constraints, \Alg{ATBS} was run on sample permutations for lengths $n \in [12, 206]$.
On an Intel i7-6700HQ processor with a clock speed of 2.60 GHz, OES took about 0.04 seconds to implement each permutation of length 512; \Alg{TBS} took about 0.3 seconds;
and, for permutations of length 200, \Alg{ATBS} took about 6 seconds.

The results of our experiments are summarized in \cref{fig:asymp-mean}.
We find that the mean normalized time costs for OES, \Alg{TBS}, and \Alg{ATBS} are similar for small $n$, but the latter two decrease steadily as the lengths of the permutations increase while the former steadily increases.
Furthermore, the average costs for \Alg{TBS} and \Alg{ATBS} diverge from that of OES rather quickly, suggesting that \Alg{TBS} and \Alg{ATBS} perform better on average for somewhat small permutations ($n \approx 50$) as well as asymptotically.

The linear coefficient $a$ of the fit of $\mu_n$ for OES is $a \approx 0.9999 \approx 1$, which is consistent with the asymptotic bound proven in \cref{avg-OES,avg-TBS}.
For the fit of the mean time costs for \Alg{TBS} and \Alg{ATBS}, we have $a \approx 0.6599$ and $a \approx 0.6513$ respectively.
The numerics suggest that the algorithm routing times agree with our analytics, and are fast for instances of realistic size. For example, at $n=100$, \Alg{TBS} and \Alg{ATBS} have routing times of $\sim 0.75 n$ and $0.72 n$, respectively. On the other hand, OES routes in average time $> 0.9n$. For larger instances, the speedup approaches the full factor of $2/3$ monotonically. 
Moreover, the fits of the standard deviations suggest $\sigma_n/n = \Theta(1/\sqrt{n})$ asymptotically, which implies that as permutation length increases, the distribution of routing times gets relatively tighter for all three algorithms. This suggests that the average-case routing time may indeed be representative of typical performance for our algorithms for permutations selected uniformly at random.

%%%%%%%%%%%%%%%%%%%%%%%%%%%%%%%%%%%%%%%%%%%%%%%%%%%%%%%%%%%%%%%%%%%%%%%%%%%%%%%%
\section{Conclusion}\label{conclusion}

We have shown that our algorithm, \Alg{ATBS} (i.e., Generic Divide-and-Conquer with Adaptive TBS to sort binary strings), uses the fast state reversal primitive to outperform any \swap{}-based protocol when routing on the path in the worst and average case.
Recent work shows a lower bound on the time to perform a reversal on the path graph of $n/\alpha$, where $\alpha \approx 4.5$~\cite{spin-chain}.
Thus we know that the routing time cannot be improved by more than a factor $\alpha$ over \swap{}s,
even with new techniques for implementing reversals.
However, it remains to understand the fastest possible routing time on the path.
Clearly, this is also lower bounded by $n/\alpha$.
Our work could be improved by addressing the following two open questions:
\begin{ienumerate} \item how fast can state reversal be implemented, and \item what is the fastest way of implementing a general permutation using state reversal? \end{ienumerate}

We believe that the upper bound in \cref{permutationbound} can likely be decreased.
For example, in the proof of \cref{TBSboundlemma}, we use a simple bound to show that the reversal sequence found by \Alg{TBS} sorts binary strings with fewer than $cn$ ones sufficiently fast for our purposes.
It is possible that this bound can be decreased if we consider the reversal sequence found by \Alg{ATBS} instead.
Additionally, in the proof of \cref{bitstringbound}, we only consider two pairs of partition points: one pair in each case of the proof.
This suggests that the bound in \cref{bitstringbound} might be decreased if the full power of \Alg{ATBS} could be analyzed.

Improving the algorithm itself is also a potential avenue to decrease the upper bound in \cref{permutationbound}. 
For example, the generic divide-and-conquer approach in \cref{alg:divconqgeneric} focused on splitting the path exactly in half and recursing.
An obvious improvement would be to create an adaptive version of \cref{alg:divconqgeneric} in a manner similar to \Alg{ATBS} where instead of splitting the path in half, the partition point would be placed in the optimal spot.
It is also possible that by going beyond the divide-and-conquer approach, we could find faster reversal sequences and reduce the upper bound even further.
% Indeed, our reversal-based algorithm also seems to obtain a significant constant factor speedup in the average case over odd-even sort, as evident by \cref{avg-TBS} and our numerical results (\cref{fig:asymp-mean}). 

Our algorithm uses reversals to show the first quantum speedup for unitary quantum routing.
It would be interesting to find other ways of implementing fast quantum routing that are not necessarily based on reversals.
Other primitives for rapidly routing quantum information might be combined with classical strategies to develop fast general-purpose routing algorithms, possibly with an asymptotic scaling advantage.
Such primitives might also take advantage of other resources, such as long-range Hamiltonians or the assistance of entanglement and fast classical communication.

\section*{Acknowledgements}
We thank William Gasarch for organizing the REU-CAAR program that made this project possible.

A.B.\ and A.V.G.\ acknowledge support by the DoE ASCR Quantum Testbed Pathfinder program (award number de-sc0019040), ARO MURI, DoE ASCR Accelerated Research in Quantum Computing program (award number de-sc0020312), U.S. Department of Energy award number de-sc0019449, NSF PFCQC program, AFOSR, and AFOSR MURI.
A.M.C.\ and E.S.\ acknowledge support by the U.S.\ Department of Energy, Office of Science, Office of Advanced Scientific Computing Research, Quantum Testbed Pathfinder program (award number de-sc0019040) and the U.S.\ Army Research Office (MURI award number W911NF-16-1-0349).
S.K.\ and H.S.\ acknowledge support from an NSF REU grant, REU-CAAR (CNS-1952352).
E.S.\ acknowledges support from an IBM Ph.D.\ Fellowship.

\printbibliography%

\appendix
%%%%%%%%%%%%%%%%%%%%%%%%%%%%%%%%%%%%%%%%%%%%%%%%%%%%%%%%%%%%%%%%%%%%%%%%%%%%%%%%
\section{Average routing time using only \swap{}s} \label{OES-appendix}

In this section, we prove \cref{avg-OES}.
First, define the infinity distance $d_{\infty}\colon \mathcal{S}_{n} \to \mathbb{N}$ to be $d_{\infty}(\pi) = \max_{1 \leq i \leq n}|\pi_{i} - i|$.
Note that $0 \leq d_{\infty}(\pi) \leq n-1$. 
Finally, define the set of permutations of length $n$ with infinity distance at most $k$ to be $B_{k,n} = \{\pi \in \mathcal{S}_{n} : d_{\infty}(\pi) \leq k\}$. 

The infinity distance is crucially tied to the performance of odd-even sort, and indeed, any \swap{}-based routing algorithm. 
For any permutation $\pi$ of length $n$, the routing time of any \swap{}-based algorithm is bounded below by $d_{\infty}(\pi)$, since the element furthest from its destination must be swapped at least $d_{\infty}(\pi)$ times, and each of those \swap{}s must occur sequentially. 
To show that the average routing time of any \swap{}-based protocol is asymptotically at least $n$, we first show that $|B_{(1-\e) n, n}|/n! \to 0$ for all $0 < \e \leq 1/2$.

Schwartz and Vontobel \cite{permutation-balls} present an upper bound on $|B_{k,n}|$ that was proved in \cite{ball-bound-rho-small} and \cite{ball-bound-rho-big}:

\begin{lemma}\label{ballsizeupperbound}
    For all $0 < r < 1$, $|B_{r n, n}| \leq \Phi(r n, n)$, where 
    \begin{align}
        \Phi(k, n) &=  
                \begin{cases*}
                    ((2k +1)!)^{\frac{n-2k}{2k+1}} \prod_{i = k+1}^{2k}(i!)^{2/i} & if ~ $0 < k/n \leq \frac{1}{2}$  \\
                     (n!)^{\frac{2k + 2 - n}{n}} \prod_{i = k+1}^{n-1}(i!)^{2/i} & if ~ $\frac{1}{2} \leq k/n < 1$.
                 \end{cases*} \label{ballsizeupperbound-expression}
    \end{align}    
\end{lemma}
\begin{proof}
    Note that $r = k/n$.
    For the case of $0 < r \leq 1/2$, refer to \cite{ball-bound-rho-small} for a proof. 
    For the case of $1/2 \leq r < 1$, refer to \cite{ball-bound-rho-big} for a proof.
\end{proof}

% \begin{lemma}\label{factorial-bound}
% For all $n = 1,2,3,...$, 
%     \begin{align}
%         \sqrt{2\pi n}\left(\frac{n}{e}\right)^{n} \leq n! \leq \sqrt{2\pi n}\left(\frac{n}{e}\right)^{n}e. 
%     \end{align}
% \end{lemma}
% \begin{proof}
% This result follows from Robbins' precise bounds for Stirling's formula \cite{robbins1955remark}: 
% \begin{align}
%     \sqrt{2\pi n}\left(\frac{n}{e}\right)^{n}e^{\frac{1}{12n+1}} &\leq n! \leq \sqrt{2\pi n}\left(\frac{n}{e}\right)^{n}e^{\frac{1}{12n}} \\
%     \sqrt{2\pi n}\left(\frac{n}{e}\right)^{n} &\leq n! \leq \sqrt{2\pi n}\left(\frac{n}{e}\right)^{n}e.
% \end{align}
% \end{proof}

\begin{lemma}\label{factorial-bound}
    \begin{align}
        n! = \Theta \left(\sqrt{n}\left(\frac{n}{e}\right)^{n}\right)
    \end{align}
\end{lemma}
\begin{proof}
This follows from well-known precise bounds for Stirling's formula: 
\begin{align}
    \sqrt{2\pi n}\left(\frac{n}{e}\right)^{n}e^{\frac{1}{12n+1}} &\leq n! \leq \sqrt{2\pi n}\left(\frac{n}{e}\right)^{n}e^{\frac{1}{12n}} \\
    \sqrt{2\pi n}\left(\frac{n}{e}\right)^{n} &\leq n! \leq \sqrt{2\pi n}\left(\frac{n}{e}\right)^{n}e
\end{align}
(see for example \cite{robbins1955remark}).
\end{proof}

With \cref{ballsizeupperbound,factorial-bound} in hand, we proceed with the following theorem:

\begin{theorem}\label{spherevanishesproof}
    For all $0 < \e \leq 1/2$, $\lim_{n \to \infty}|B_{(1-\e)n, n}|/n! = 0$. 
    In other words, the proportion of permutations of length $n$ with infinity distance less than $(1-\epsilon)n$ vanishes asymptotically.
\end{theorem}

\begin{proof}
\cref{ballsizeupperbound} implies that $|B_{(1-\e)n, n}|/n! \le \Phi((1-\e)n, n)/n!$. The constraint $0 < \e \leq 1/2$ stipulates that we are in the regime where $1/2 \leq r < 1$, since $r = 1 - \e$. Then we use \cref{factorial-bound} to simplify any factorials that appear.
Substituting \cref{ballsizeupperbound-expression} and simplifying, we have 
\begin{align}
    \frac{\Phi\left((1-\e)n, n\right)}{n!} &= \frac{\prod_{i = (1-\e)n + 1}^{n-1}(i!)^{2/i}}{(n!)^{2\e - 2/n}} =
    O\left(\frac{e^{2\e n - 2}}{n^{2 \e n - 2}}\prod_{i = (1-\e)n + 1}^{n-1}\frac{i^{2+1/i}}{e^{2}}\right).
\end{align}
We note that $i^{1/i}$ terms can be bounded by
\begin{equation}
    \prod_{i=(1-\e)n+1}^{n-1} i^{\frac{1}{i}} \le \prod_{i=(1-\e)n+1}^{n-1} n^{\frac{1}{(1-\e)n}} \le n^{\frac{\e}{1-\e}} \le n
\end{equation}
since $\e \le 1/2$.
Now we have
\begin{align}
    O\left(\frac{e^{2\e n - 2}}{n^{2 \e n - 2}}\prod_{i = (1-\e)n + 1}^{n-1}\frac{i^{2+1/i}}{e^{2}}\right)
    &= O\left(\frac{n}{n^{2 \e n - 2}}\prod_{i = (1-\e)n + 1}^{n-1}i^{2}\right) \\
    &= O\left(\frac{n}{n^{2 \e n - 2}}\left(\frac{(n-1)!}{((1-\e)n +1)!}\right)^{2}\right) \\
    &= O\left(\frac{n}{n^{2 \e n - 2}e^{2\e n}}\frac{(n-1)^{2n-1}}{((1-\e)n+1)^{2(1-\e)n+2}}\right) \\
    &= O\left(\frac{n}{n^{2 \e n - 2}e^{2\e n}}\frac{n^{2n}}{((1-\e)n)^{2(1-\e)n}}\right) \\
    &= \bigo*{\frac{n^{3}}{\exp\left((\ln(1-\e)(1-\e) + \e)2n\right)}}.\label{eq:show-const-1}
\end{align}
Since $\ln(1-\e)(1-\e) + \e > 0$ for $\e >0$, this vanishes in the limit of large $n$.
\end{proof}

Now we prove the theorem.
\begin{proof}[Proof of~\cref{avg-OES}]  
    Let $\Bar{T}$ denote the average routing time of any \swap{}-based protocol.
    Consider a random permutation $\pi$ drawn uniformly from $\mathcal{S}_{n}$. 
    Due to \cref{spherevanishesproof}, $\pi$ will belong in $B_{(1-\e)n, n}$ with vanishing probability, for all  $0 < \e \leq 1/2$.  Therefore, for any fixed $0 < \e \leq 1/2$ as $n \to \infty$, $(1-\e)n < \mathbb{E}\left[d_{\infty}(\pi)\right]$. 
    This translates to an average routing time of at least $n - o(n)$ because we have, asymptotically, $(1-\e)n \leq \Bar{T}$ for all such $\e$.
\end{proof}

\newcommand{\tbs}[0]{T}

\section{Average routing time using TBS}\label{TBS-appendix}

In this section, we prove~\cref{avg-TBS}, which characterizes the average-case performance of TBS (\cref{alg:TBS}).
This approach consists of two steps: a recursive call on three equal partitions of the path (of length $n/3$ each), and a merge step involving a single reversal.

We denote the uniform distribution over a set $S$ as $\mathcal{U}(S)$. The set of all $n$-bit strings is denoted $\mathbb{B}^n$, where $\mathbb{B}=\{0,1\}$. Similarly, the set of all $n$-bit strings with Hamming weight $k$ is denoted $\mathbb{B}_k^n$. For simplicity, assume that $n$ is even. We denote the runtime of TBS on $b \in \mathbb{B}^n$ by $T(b)$.
% The runtime of the TBS algorithm on input permutation $\pi$ will be denoted $\tbs(\pi)$. The input to TBS is not a permutation but the corresponding 0-1 sequence of length $n$, which we denote $b(\pi)$. The $i$-th bit of this sequence is given by $b(\pi)_i = \pi(i)\pmod{n/2}$. We will sometimes write the runtime of TBS on the 0-1 string $b(\pi)$ as $\tbs(b)$, which is the same as $\tbs(\pi)$.

When running $\Alg{TBS}$ on a given permutation $\pi$, the input bit string for TBS is $b = I(\pi)$, where the indicator function $I$ is defined in \cref{eq:indicator}.
We wish to show that, in expectation over all permutations $\pi$, the corresponding bit strings are quick to sort.
First, we show that it suffices to consider uniformly random sequences from $\mathbb{B}_{n/2}^n$.

\begin{lemma}
  \label{lem:bitization}
If $\pi\sim \mathcal{U}(\mathcal{S}_n)$, then 
% the binary mapping is drawn according to 
$I(\pi)\sim\mathcal{U}(\mathbb{B}_{n/2}^n)$.
\end{lemma}

\begin{proof}
We use a counting argument. The number of permutations $\pi$ such that $I(\pi)\in\mathbb{B}_{n/2}^n$ is $(n/2)!(n/2)!$, since we can freely assign index labels from $\{1,2,\ldots, n/2\}$ to the 0 bits of $I(\pi)$, and from $\{n/2+1,\ldots, n\}$ to the 1 bits of $I(\pi)$. Therefore, for a uniformly random $\pi$ and arbitrary $b\in\mathbb{B}_{n/2}^n$, 
\begin{equation}
    \probability*{I(\pi) = b} = \frac{(n/2)!(n/2)!}{n!} = \frac{1}{\binom{n}{n/2}} = \frac{1}{|\mathbb{B}_{n/2}^n|}.
\end{equation}
Therefore, $I(\pi)\sim\mathcal{U}(\mathbb{B}_{n/2}^n)$.
\end{proof}

While $\mathbb{B}_{n/2}^{n}$ is easier to work with than $\mathcal{S}_n$, the constraint on the Hamming weight still poses an issue when we try to analyze the runtime recursively. To address this, \cref{lem:relaxation} below shows that relaxing from $\mathcal{U}(\mathbb B_{n/2}^n)$ to $\mathcal{U}(\mathbb{B}^n)$ does not affect expectation values significantly.

We give a recursive form for the runtime of TBS. We use the following convention for the substrings of an arbitrary $n$-bit string $a$: if $a$ is divided into $3$ segments, we label the segments $a_{0.0},a_{0.1},a_{0.2}$ from left to right. Subsequent thirds are labeled analogously by ternary fractions. For example, the leftmost third of the middle third is denoted $a_{0.10}$, and so on. Then, the runtime of TBS on string $a$ can be bounded by
\begin{equation}
    \tbs(a) \le \max_{i\in \set{0,1,2}} \tbs(a_{0.i}) + \frac{n_1(a_{0.0}) + n_1(\overline{a_{0.2})} + n/3 + 1}{3},
\end{equation}
where $\overline{a}$ is the bitwise complement of bit string $a$ and $n_1(a)$ denotes the Hamming weight of $a$. Logically, the first term on the right-hand side is a recursive call to sort the thirds, while the second term is the time taken to merge the sorted subsequences on the thirds using a reversal. Each term $T(a_{0.i})$ can be broken down recursively until all subsequences are of length 1. This yields the general formula
\begin{equation}
\label{eq:TBS-recursion}
  \tbs(b) \le \frac{1}{3}\left(\sum\limits_{r=1}^{\lceil\log_3(n)\rceil}\max_{i\in \set{0,1,2}^{r-1}}\{n_1(a_{0.i0}) + n_1(\overline{a_{0.i2}})\} + n/3^r + 1\right),
\end{equation}
where $i \in \emptyset$ indicates the empty string.

\begin{lemma}
  \label{lem:relaxation}
  Let $a\sim\mathcal{U}(\mathbb{B}^n)$ and $b\sim\mathcal{U}(\mathbb{B}_{n/2}^n)$. Then
  \begin{equation}
    \expected*{\tbs(b)} \le \expected*{\tbs(a)} + \widetilde{O}(n^\alpha)
  \end{equation}
where $\alpha\in(\frac{1}{2},1)$ is a constant.
\end{lemma}

%\begin{theorem}
%  \label{avg-TBS}
 % Let $\hat{\pi}\sim\mathcal{U}(\mathcal{S}_n)$. Then, $\mathbb{E}(\tbs(\hat{\pi}))\le \frac{n}{3} + o(n)$.
%\end{theorem}

The intuition behind this lemma is that by the law of large numbers, the deviation of the Hamming weight from $n/2$ is subleading in $n$, and the TBS runtime does not change significantly if the input string is altered in a subleading number of places.

\begin{proof}
  Consider an arbitrary bit string $a$, and apply the following transformation. If $n_1(a) = k \ge n/2$, then flip $k-n/2$ ones chosen uniformly randomly to zero. If $k < n/2$, flip $n/2-k$ zeros to ones. Call this stochastic function $f(a)$. Then, for all $a$, $f(a)\in \mathbb{B}_{n/2}^n$, and for a random string $a \sim \mathcal{U}(\mathbb{B}^n)$, we claim that $f(a) \sim \mathcal{U}(\mathbb{B}_{n/2}^n)$. In other words, $f$ maps the uniform distribution on $\mathbb{B}^n$ to the uniform distribution on $\mathbb{B}_{n/2}^n$.
  
  We show this by calculating the probability $\probability*{f(a)=b}$, for arbitrary $b\in\mathbb{B}_{n/2}^n$. A string $a$ can map to $b$ under $f$ only if $a$ and $b$ disagree in the same direction: if, WLOG, $n_1(a)\ge n_1(b)$, then $a$ must take value 1 wherever $a,b$ disagree (and 0 if $n_1(a)\le n_1(b)$). We denote this property by $a\succeq b$. The probability of picking a uniformly random $a$ such that $a\succeq b$ with $x$ disagreements between them is $\binom{n/2}{x}$, since $n_0(b)=n/2$. Next, the probability that $f$ maps $a$ to $b$ is $\binom{n/2+x}{x}$. Combining these, we have 
    \begin{align}
    \probability*{f(a) = b} &= \sum\limits_{x=-n/2}^{n/2} \probability*{a\succeq b \text{ and } n_1(a)=\frac{n}{2}+x} \cdot \probability*{f(a)=b \mid a\succeq b \text{ and } n_1(a)=\frac{n}{2}+x} , \\
    &= \sum\limits_{x=-n/2}^{n/2}\frac{\binom{n/2}{|x|}}{2^n}\cdot\frac{1}{\binom{n/2+|x|}{|x|}} ,\\ 
    &= \frac{1}{\binom{n}{n/2}}\sum\limits_{x=-n/2}^{n/2}\frac{\binom{n}{n/2 - x}}{2^n} ,\\ 
    &= \frac{1}{\binom{n}{n/2}} = \frac{1}{|\mathbb{B}_{n/2}^n|} .
  \end{align}
Therefore, $f(a)\sim \mathcal{U}(\mathbb{B}_{n/2}^n)$. Thus, $f$ allows us to simulate the uniform distribution on $\mathbb{B}_{n/2}^n$ starting from the uniform distribution on $\mathbb{B}^n$.
  
Now we bound the runtime of TBS on $f(a)$ in terms of the runtime on a fixed $a$.
Fix some $\alpha\in(\frac{1}{2},1)$.
We know that $n_1(f(a)) = n/2$,
and suppose $\abs{n_1(a) - n/2} \le n^\alpha$.
Since $f(a)$ differs from $a$ in at most $n^\alpha$ places, then at level  $r$ of the TBS recursion (see~\cref{eq:TBS-recursion}), the runtimes of $a$ and $f(a)$ differ by at most $1/3\cdot\min \{2n/3^r, n^\alpha\}$.
This is because the runtimes can differ by at most two times the length of the subsequence.
Therefore, the total runtime difference is bounded by
\begin{align}
  \Delta \tbs &\le 
  \frac{1}{3}\sum\limits_{r=1}^{\lceil\log_3(n)\rceil} \min\Bigl\{\frac{2n}{3^r}, n^\alpha\Bigr\} ,\\ 
  &=\frac{1}{3}\left(\sum\limits_{r=1}^{\lceil\log_3(2n^{1-\alpha})\rceil} n^\alpha + 2\sum\limits_{r=\lceil\log_3(2n^{1-\alpha})\rceil+1}^{\lceil\log_3(n)\rceil} \frac{n}{3^r} \right) ,\\
  &=\frac{1}{3}\left(n^\alpha\log(2n^\alpha/3) + 2\sum\limits_{s=0}^{\lfloor\log_3(n^\alpha/2)\rfloor-1} 3^s\right)\\
 &= \frac{1}{3}\left(n^\alpha\log(2n^\alpha/3) + n^\alpha/2 - 1\right) = \widetilde{O}(n^{\alpha}).
\end{align}
On the other hand, if $\abs{n_1(a) - n/2} \ge n^\alpha/2$, we simply bound the runtime by that of OES, which is at most $n$.

Now consider $a \sim \mathcal{U}(\mathbb{B}^n)$ and
% Now let us consider a uniformly random permutation $\pi \sim \mathcal{U}(\mathcal S_n)$ and its associated bit string $\hat b = I(\pi)$, with 
$b = f(a) \sim \mathcal{U}(\mathbb{B}^n_{n/2})$.
Since $n_1(a)$ has the binomial distribution $\mathcal B(n, 1/2)$,
where $\mathcal B(k,p)$ is the sum of $k$ Bernoulli random variables with success probability $p$,
the Chernoff bound shows that
deviation from the mean is exponentially suppressed, i.e.,
\begin{equation}
    \probability{\abs{n_1(a) - n/2}\ge n^\alpha} = \exp(-\bigo{n^{2\alpha - 1}}).
\end{equation}
Therefore, the deviation in the expectation values is bounded by 
\begin{equation}
    \abs*{\expected*{\tbs(f(a))} - \expected*{\tbs(a)}} \le n\exp(-\bigo{n^{2\alpha - 1}}) + c(1-\exp(-\bigo{n^{2\alpha - 1}}) )n^{\alpha}\log(n) = \widetilde{O}(n^{\alpha}),
\end{equation}
where $c$ is a constant. Finally, we conclude that
\begin{equation}
  \expected{\tbs(b)}\le \expected{\tbs(a)} + \widetilde{O}(n^\alpha)
\end{equation}
as claimed.
\end{proof}

Next, we prove the main result of this section, namely, that the runtime of \Alg{TBS} is $2n/3$ up to additive subleading terms.

\begin{proof}[Proof of~\cref{avg-TBS}]
We first prove properties for sorting a random $n$-bit string $a \sim \mathcal U(\mathbb B^n)$
and then apply this to the problem of sorting $b \sim \mathcal U(\mathbb B^n_{n/2})$ using \cref{lem:bitization,lem:relaxation}.

The expected runtime for TBS can be calculated using the recursive formula in~\cref{eq:TBS-recursion}:
\begin{equation}
  \expected{\tbs(a)} \le \frac{1}{3}\left(\sum\limits_{r=1}^{\log_3(n)}\expected*{\max_{i\in \set{0,1,2}^{r-1}}\{n_1(a_{0.i0}) + n_1(\overline{a_{0.i2}})\}} + n/3^r + 1\right).\label{eq:expectedTimeBitstrings}
\end{equation}
The summand contains an expectation of a maximum over Hamming weights of i.i.d.\ uniformly random substrings of length $n/3^r$, which is equivalent to a binomial distribution $\mathcal B(n/3^r, 1/2)$ where we have $n/3^r$ Bernoulli trials with success probability $1/2$.
Because of independence, if we sample $X_1, X_2  \sim \mathcal B(n/3^r,1/2)$, then $X_1 + X_2 \sim \mathcal B(2n/3^r, 1/2)$.
Using \cref{lem:expect-max} with $m=3^{r-1}$, the expected maximum can be bounded by
\begin{align}
    %  \frac{n}{3^r} + \sqrt{\frac{n (r-1)\log 3}{2\cdot 3^r}} + \min\set*{\frac{2n}{3^r}, \bigo*{3^{r/2-1}\sqrt{2n}}} = \frac{n}{3^r} + \bigo*{n^{2/3}},
    \frac{n}{3^r} + O\left(\sqrt{(n/3^r)\log(3^{r-1}n/3^r)}\right)
    = \frac{n}{3^r} + \tildebigo*{n^{1/2}}
\end{align}
% where we maximized the minimum at $r=(\log_3(2n)+2)/3$.
since the second term is largest when $r=O(1)$.
Therefore,
\begin{equation}
    \expected{\tbs(a)} \le \frac{1}{3}\left(\sum\limits_{r=1}^{\log_3(n)} \frac{2n}{3^r} \right) + \tildebigo*{n^{1/2}} = \frac{n}{3} + \tildebigo*{n^{1/2}}. 
\end{equation}
\cref{lem:relaxation} then gives
$\expected{\tbs(b)} \le \frac{n}{3} + \tildebigo*{n^\alpha}$.

The routing algorithm \Alg{TBS} proceeds by calling TBS on the full path, and then in parallel on the two disjoint sub-paths of length $n/2$.
We show that the distributions of the left and right halves are uniform if the input permutation is sampled uniformly as $\pi \sim \mathcal U(\mathcal S_n)$.
There exists a bijective mapping $g$ such that $g(\pi) = (b, \pi_L, \pi_R) \in \mathbb B^n_{n/2} \times \mathcal S_{n/2}\times \mathcal S_{n/2}$ for any $\pi \in \mathcal S_n$ since
\begin{equation}
    \abs*{\mathcal S_n} = n! =  \binom{n}{n/2} \left(\frac{n}{2}\right)! \left(\frac{n}{2}\right)! = \abs*{\mathbb B^n_{n/2} \times \mathcal S_{n/2}\times \mathcal S_{n/2}}.
\end{equation}
In particular, $g$ can be defined so that $b$ specifies which entries are taken to the first $n/2$ positions---say, without changing the relative ordering of the entries mapped to the first $n/2$ positions or the entries mapped to the last $n/2$ positions---and $\pi_L$ and $\pi_R$ specify the residual permutations on the first and last $n/2$ positions, respectively.
Given $g(\pi) = (b, \pi_L, \pi_R)$, TBS only has access to $b$.
After sorting, TBS can only perform deterministic permutations $\mu_L(b), \mu_R(b) \in \mathcal S_{n/2}$ on the left and right halves, respectively, that depend only on $b$.
Thus TBS performs the mappings $\pi_L \mapsto \pi_L \circ (\mu_L(b))$ and $\pi_R \mapsto \pi_R \circ (\mu_R(b))$ on the output.
Now it is easy to see that when $\pi_L, \pi_R \sim \mathcal U(\mathcal S_{n/2})$,
the output is also uniform because the TBS mapping is independent of the relative permutations on the left and right halves.

More generally, we see that a uniform distribution over permutations $\mathcal U(\mathcal S_n)$ is mapped to two uniform permutations on the left and right half, respectively.
Symbolically, for, $\pi \sim \mathcal U(\mathcal S_n)$,
we have that 
\begin{equation}
    g(\pi) = (b, \pi_L, \pi_R) \sim \mathcal U(\mathbb B^n_{n/2} \times \mathcal S_{n/2}\times \mathcal S_{n/2}) = \mathcal U(\mathbb B^n_{n/2}) \times \mathcal U(\mathcal S_{n/2}) \times \mathcal U(\mathcal S_{n/2}).
\end{equation}
As shown earlier, given uniform distributions over left and right permutations, the output is also uniform.
By induction, all permutations in the recursive steps are uniform.

We therefore get a sum of expected TBS runtime on bit strings of lengths $n/3^r$, i.e.,
\begin{equation}
\label{eq:GDC-recursion}     \sum_{r=1}^{\log_2 n} \expected{T(b_r)} \le \sum_{r=1}^{\log_2 n} \expected{T(a_r)} + \tildebigo*{\left(\frac{n}{2^{r-1}}\right)^\alpha} \le \frac{2n}{3} + \tildebigo{n^\alpha}
\end{equation}
where, by \cref{lem:bitization} and the uniformity of permutations in recursive calls, we need only consider $b_r \sim \mathcal U(\mathbb B_{n/2^{r-1}}^{n/2^r})$ and
we bound the expected runtime using \cref{lem:relaxation} with $a_r \sim \mathcal U(\mathbb B^{n/2^{r-1}})$.
\end{proof}
    
We end with a lemma about the order statistics of binomial random variables used in the proof of the main theorem.

\begin{lemma} 
\label{lem:expect-max}
Given $m$ i.i.d.\ samples from the binomial distribution $X_i \sim \mathcal \mathcal{B}(n,p)$
with $i \in [m]$, and $p\in [0,1]$,
the maximum $Y = \max_i X_i$ satisfies
\begin{equation}
\expected Y < pn + O\left(\sqrt{n\log(mn)}\right).
\end{equation}
\end{lemma}
\begin{proof}
We use Hoeffding's inequality for the Bernoulli random variable $X\sim \mathcal{B}(n,p)$, which states that 
\begin{equation}
\probability*{X\ge (p+\epsilon)n} \le \exp(-2n\epsilon^2) \quad \forall \e \geq 0.   
\end{equation}

Pick $\epsilon = \sqrt{\frac{c}{2n}\log(mn)}$, where $c>0$ is a constant. 
For this choice, we have 
\begin{equation}
    \probability*{X_i\ge (p+\epsilon)n} \le \left(\frac{1}{mn}\right)^c
\end{equation} 
for every $i=1,\ldots,m$. 
Then the probability that $Y < (p+\epsilon)n$ is identical to the probability that $X_i < (p+\epsilon)n$ for every $i$, which for i.i.d $X_i$ is given by
\begin{equation}
\probability*{Y < (p+\epsilon)n} = \probability*{X < (p+\epsilon)n}^m > \left(1-\frac{1}{(mn)^c}\right)^m.
\end{equation}
Using Bernoulli's inequality ($(1+x)^r \ge 1+rx$ for $x\ge -1$), we can simplify the above bound to
\begin{equation}
  \probability*{Y < (p+\epsilon)n}^m > 1-m^{1-c}n^{-c}.
\end{equation}
Finally, we bound the expected value of $Y$ by an explicit weighted sum over its range:
\begin{align}
    \expected Y &= \sum_{k=0}^n \probability*{Y = k}\cdot k \\
    &= \sum_{k=0}^{\lfloor(p+\epsilon)n\rfloor} \probability*{Y = k}\cdot k + \sum_{k=\lfloor(p+\epsilon)n\rfloor+1}^{n} \probability*{Y = k}\cdot k \\
    &\leq \sum_{k=0}^{\lfloor(p+\epsilon)n\rfloor} \probability*{Y=k})\cdot k + n \cdot \sum_{k=\lfloor(p+\epsilon)n\rfloor+1}^{n} \probability*{Y = k} \\
    &\leq \sum_{k=0}^{\lfloor(p+\epsilon)n\rfloor} \probability*{Y = k}\cdot k + (mn)^{1-c} \\
    &\leq (p+\epsilon)n + (mn)^{1-c}.
\end{align}
Since $(mn)^{1-c} < 1$ for $c>1$,
\begin{equation}
   \expected Y < \left\lceil pn + 1 + \sqrt{\frac{cn}{2}\log(mn)} \, \right\rceil = pn + O(\sqrt{n\log(mn)})\\
\end{equation}
as claimed.
\end{proof}

\end{document}